\documentclass[a4paper,11pt]{amsart}

\usepackage[hidelinks]{hyperref}

\usepackage{amsmath,amsthm,amssymb,latexsym,amsfonts,xcolor}
\usepackage[latin1]{inputenc}
\usepackage[english,activeacute]{babel}
\usepackage{graphicx}

\usepackage{enumerate}
\usepackage{csquotes}

\usepackage{parskip}
\usepackage{hyperref}
\usepackage{cleveref}

\usepackage{cancel}
\usepackage[normalem]{ulem}
\definecolor{coralred}{rgb}{1.0, 0.25, 0.25}

\DeclareMathOperator{\sinc}{sinc}

\DeclareMathOperator{\vdiv}{div}
\newcommand{\norm}[1]{\left\lVert#1\right\rVert}
\allowdisplaybreaks

\def \N{\mathbb N}
\def \Z{\mathbb Z}

\def \R{\mathbb R}
\def \C{\mathbb C}

\def \pa{{\partial}}
\def \Sc{\mathcal{S}}

\def \Fc{\mathcal{F}}

\def \E{\mathbb{E}}

\def\build#1_#2^#3{\mathrel{\mathop{\kern 0pt#1}\limits_{#2}^{#3}}}

\numberwithin{equation}{section}

\theoremstyle{plain}
\newtheorem{thm}{Theorem}[section]

\newtheorem{lem}[thm]{Lemma}
\newtheorem{prop}[thm]{Proposition}
\newtheorem{cor}[thm]{Corollary}
\newtheorem{rk}[thm]{Remark}

\theoremstyle{definition}
\newtheorem{defn}[thm]{Definition}

\theoremstyle{remark}

\theoremstyle{plain}

\title[A linear model of turbulent cascades]{A linear stochastic model of turbulent cascades and fractional fields}

\author[G. B. Apolin\'ario]{Gabriel B. Apolin\'ario}
\address{Theoretical Physics I, University of Bayreuth, Universit\"atsstr. 30, 95447 Bayreuth, Germany}
\email{gapolinario@uni-bayreuth.de} 

\author[G. Beck]{Geoffrey Beck}
\address{Univ Rennes, IRMAR UMR 6625 \& Centre Inria de l'Universit\'e de Rennes (MINGuS) \& ENS Rennes, France}
\email{geoffrey.a.beck@inria.fr} 

\author[L. Chevillard]{Laurent Chevillard}
\address{Univ Lyon, ENS de Lyon, Univ. Claude Bernard, CNRS, Laboratoire de Physique, 46 all\'ee d'Italie, 69342 Lyon, France \& Univ Lyon, Universit\'e Claude Bernard Lyon 1, CNRS UMR 5208, Institut Camille Jordan, 43 boulevard du 11 Novembre 1918, F-69622 Villeurbanne, France}
\email{laurent.chevillard@cnrs.fr} 

\author[I. Gallagher]{Isabelle Gallagher}
\address{D\'epartement de math\'ematiques et applications, \'Ecole normale sup\'erieure, CNRS, PSL University,  and Universit\'e Paris Cit\'e, 75005 Paris, France}
\email{isabelle.gallagher@ens.fr} 

\author[R. Grande]{Ricardo Grande}
\address{D\'epartement de math\'ematiques et applications, \'Ecole normale sup\'erieure, CNRS, PSL University, 75005 Paris, France}
\email{ricardo.grande@ens.fr}

\begin{document}
	
\begin{abstract}

Turbulent cascades characterize the transfer of energy injected by a random force at large scales towards the small scales. In hydrodynamic turbulence, when the Reynolds number is large, the velocity field of the fluid becomes irregular and the rate of energy dissipation remains bounded from below even if the fluid viscosity tends to zero. A mathematical description of the turbulent cascade is a very active research topic since the pioneering work of Kolmogorov in hydrodynamic turbulence and that of Zakharov in wave turbulence. In both cases, these turbulent cascade mechanisms imply power-law behaviors of several statistical quantities such as power spectral densities. For a long time, these cascades were believed to be associated with nonlinear interactions, but  recent works have shown that they can also take place in a dynamics governed by a linear equation with a pseudo-differential operator of degree 0. In this spirit, we construct a linear equation that mimics the phenomenology of energy cascades when the external force is a statistically homogeneous and stationary stochastic process. In the Fourier variable, this equation can be seen as a linear transport equation, which corresponds to an operator of degree 0 in physical space. Our results give a complete characterization of the solution: it is smooth at any finite time, and, up to smaller order corrections, it converges to a fractional Gaussian field at infinite time.
\end{abstract}

	\maketitle
	
\section{Introduction}

\subsection{Background and motivation}

This work is mainly motivated by some important aspects of the phenomenology of three-dimensional homogenous and isotropic fluid turbulence \cite{monin1971statistical,tennekes1972first,frisch1995turbulence}, of which several aspects have been also observed and formalized for waves in various situations  when they are weakly interacting \cite{zakharov2012kolmogorov}.  As has been repeatedly observed in geophysical and laboratory flows, and in numerical simulations of the incompressible Navier-Stokes equations, a fluid that is stirred by a statistically stationary random force $f(t,x)$, assumed to be smooth in space, will eventually reach a statistically stationary state in which the velocity variance is finite. To dissipate all the energy that is constantly injected into the system in such an efficient way, the velocity field of that fluid will develop a complex multiscale structure ending up with high values of   velocity gradients such that viscosity can easily transform mechanical energy into heat. In other words, the fluid has transferred the energy pumped at large scales by the forcing towards small scales, at which viscous diffusion efficiently acts. This picture is known as the cascading process of energy.

The purpose of this article is to model and reproduce this phenomenon of transfer of energy as a cascading process through the scales. We propose a partial differential equation, which is of course much simpler than the nonlinear Navier-Stokes equations, stochastically forced by an additive random force $f(t,x)$ that we take to be smooth in space and correlated over a typical large lengthscale  (known in the turbulence literature as the integral lengthscale), whose solution develops roughness as time goes on. More precisely, our goal is to generate rough fractional Gaussian H\"{o}lder continuous random fields of parameter $H$ (see for instance the textbook \cite{cohen2013fractional}) from smooth forcing through a dynamical evolution, which can be seen as a simple stochastic representation of the phenomenology mainly developed by Kolmogorov  \cite{1941DoSSR..32...16K}.

As mentioned earlier, a striking feature of three-dimensional turbulent motion is its ability to efficiently dissipate the energy that is injected at large scales in a statistically stationary and homogeneous manner. To be more precise, let us consider a solution of the incompressible Navier-Stokes equation, i.e. a divergence-free velocity field $u(t,x)\in\R^3$ with periodic boundary conditions. This dynamics is stirred by a divergence-free vector forcing term $f(t,x)$, that we take delta-correlated in time and smooth in space, say Gaussian, of zero-average and of covariance
\[
\E \left[f(t,x)\otimes f(s,y)\right] = \delta_{t-s} C_f(x-y),
\]
where $\otimes$ stands for the matrix product, and the matrix $C_f(x)$ is made of a linear combination of the matrix $x\otimes x$ and the identity, with multiplicative coefficients depending only on $|x|$, i.e. a typical covariance matrix of a statistically homogeneous and isotropic vector field \cite{batchelor1953theory,pope2000turbulent}. We furthermore require that these scalar functions of $|x|$ are smooth and compactly supported over a range of the size order of the aforementioned large length scale, so as to mimic the energy injection at the so-called integral length scale. As time goes on, it has been repeatedly observed that the velocity field $u$ reaches a statistically stationary state, which is furthermore statistically homogeneous, of finite variance, and with the additional striking property that it becomes independent of the viscosity $\nu$ as $\nu$ goes to zero, i.e.
\begin{equation}\label{eq:AsymptVarianceNS} 
\quad \lim_{\nu\to 0}\lim_{t\to \infty}\E\left[ |u(t,x)|^2\right] <+\infty \quad \text{for all} \,  \,x.
\end{equation}
The former asymptotic behavior of the velocity variance illustrates clearly how a turbulent fluid can dissipate energy with  high efficiency. For instance, in the same setup but considering the heat equation instead of the Navier-Stokes equations, a statistically stationary regime would also be reached at $t \to \infty$. However, the variance of the solution is then inversely proportional to the viscosity $\nu$, see \cite{chevillard:tel-01212057}. Instead, turbulent motion dissipates energy in a way that the velocity variance is eventually independent of viscosity, which is a far more efficient way of dissipating energy. In order to ensure the independence of said variance on viscosity, \eqref{eq:AsymptVarianceNS}, the fluid develops a rough behavior of H\"{o}lder-type at small scales, in such a way that the variance of the velocity increments asymptotically behaves as follows:
\begin{equation}\label{eq:AsymptVarianceIncrNS} 
\lim_{\nu\to 0}\lim_{t\to \infty}\E\left[ |u(t,x+\ell)-u(t,x)|^2\right] \build{\propto}_{|\ell|\to 0}^{}|\ell|^{2H}  \quad \text{for all}  \,\,  x,
\end{equation}
where the power-law exponent is determined by Kolmogorov's prediction $H\approx 1/3$ \cite{frisch1995turbulence}. Much more could be said on a more precise characterization of the distribution of the increments than only its variance, \eqref{eq:AsymptVarianceIncrNS}, such as its higher order moments that quantify its non Gaussian, skewed and intermittent nature \cite{frisch1995turbulence}. In this article, we will focus on a second-order modeling of these fluctuations, \eqref{eq:AsymptVarianceNS} and \eqref{eq:AsymptVarianceIncrNS}; we leave finer descriptions for future research.

A first precise formalization of the cascade phenomenon could be built by imposing a particular dynamical relation between the coefficients of a decomposition of the velocity field, such as a continuous wavelet transform or a discrete (dyadic) decomposition on a tree \cite{daubechies1992ten}. This has been explored in the literature \cite{barbato2011smooth,barbato2013dyadic,cheskidov2022dyadic} leading to precise statements on H\"{o}lder regularity and its relationship with scaling behaviors of the coefficients.
Although great progress has been made in the understanding of such models and their formalization, which usually exploits a typical quadratic interaction between neighboring coefficients, these approaches often avoid the important question of the relation of these coefficients in space. This is necessary in order to design a model that leads to statistically homogeneous velocity fields, as observed in nature and in numerical simulations. Nonetheless, these models can be seen as a sophistication of the so-called shell models\footnote{See for instance \cite{bohr1998dynamical,biferale2003shell} which consist in exploring quadratic interactions between shells, that share some behaviors with velocity Fourier modes and wavelet coefficients, along a single branch of a tree decomposition, lacking thus a discussion of the spatial relationships between coefficients.}. In this spirit, we believe an important step was made in \cite{mattingly2007simple}, where the authors investigate a simple linear relation between shells, which is shown to be able to transfer energy from large to small scales. Let us also mention \cite{mailybaev2015continuous} where some ideas to build a PDE from these shell models are proposed.

From a somewhat different side of fluid mechanics, more focused on the implications of global rotation \cite{rieutord1997inertial,rieutord2001inertial} or stratification of the density field \cite{maas1997observation,scolan2013nonlinear} on a flow, it has been evidenced a phenomenon of focusing of waves onto attractors, whose precise shape are determined by the boundaries. Based on a linearization of the fluid equations, this phenomenon has been interpreted as a cascading process through scales. These ideas have been then formalized and rigorously studied from a mathematical viewpoint in a series of recent articles \cite{CSR,dyatlov2019microlocal}, which underline the importance of operators of degree 0 as a deterministic mechanism able to transfer energy through scales. 

\subsection{Main results}\label{Sec:MainResults}

The rough and disordered nature of a turbulent velocity field $u(t,x)$ has been repeatedly observed in laboratory and numerical flows, and in geophysical situations \cite{tennekes1972first,frisch1995turbulence}. From this signal, considering for instance a component of the velocity vector field as a function of space, depending on the experimental possibilities and the large-scale geometry of the flows, one can construct the energy spectrum $|k| \mapsto \mathbb{E} | \widehat{u}(t,k) |^2$ where $\widehat{u}$ stands for the spatial Fourier transform. According to the standard phenomenology of fluid turbulence, which has been multiply confirmed by   observations in very different situations, the energy-spectrum resembles a curve \cite{tennekes1972first,frisch1995turbulence} that can be schematically decomposed as follows:
\begin{itemize}
\item {\it{(injection range)}} for small $|k|$ of the order of the characteristic wavelength of energy injection, the energy-spectrum is mainly determined by the forcing and the associated large-scale geometry of the flow,
\item {\it{(inertial range)}} for intermediate $|k|$, the energy-spectrum develops a power-law behavior whose exponent is found universal, i.e. independent of viscosity and of the nature of the flow, and can be interpreted as the generation of small scales by the internal motion of the fluid following a transfer of energy from small wave-numbers to large wave-numbers,
\item {\it{(dissipative range)}} for large $|k|$, the energy-spectrum is governed by dissipation processes which damp efficiently all the energy coming from the large scales, making the spatial velocity profile   a smooth function. 
\end{itemize}
The intermediate range of scales, called the inertial range in the turbulence literature  \cite{tennekes1972first,frisch1995turbulence}, is where this mechanism of transport of energy takes place. The universally observed power-law exponent of the energy-spectrum can be written as $-(2H+d)$, i.e. $\mathbb{E} | \widehat{u}(t,k) |^2 \sim |k|^{-(2H+d)}$, where we have introduced for the sake of generality the space dimension $d$, and the parameter $H$ that will be eventually interpreted as a Hurst, or H\"{o}lder, exponent, in a statistically averaged sense. In real situations, for $d=3$, it is indeed observed that $H\approx 1/3$, as predicted by dimensional arguments mainly attributed to Kolmogorov \cite{1941DoSSR..32...16K,LewanPini}.

The main goal of this paper is to propose a family of partial differential equations, such that, when stirred by a statistically stationary and spatially homogenous, smooth in space forcing term, its solution $u(t,x)$ reaches  at long times a statistically stationary state which displays the typical spectral behavior detailed above. We will achieve this with the following transport equation in Fourier space:
\begin{equation}\label{eq:maineq-div-intro}
\begin{cases}
\pa_t \widehat{u}(t,k) +  \vdiv_k \left(  \frac{c k}{| k |} \, \widehat{u}(t,k) \right) + c \displaystyle \frac{H+ \frac{1}{2}}{|k|} \,  \widehat{u}(t,k) = \widehat{f}(t,k) & t>0, k \in \mathbb{R}^d,  |k| >\kappa >0,\\
\widehat{u}(t, k)=0 &  t>0, k \in \mathbb{R}^d, |k| \leq \kappa, \\
\widehat{u}(0,k)=0.
\end{cases}
\end{equation}
Here $c,\kappa>0$ and $H\in\R$ are fixed, and the source $f$ satisfies
$$
\E [ f(t,x) f(s,y)] = \delta_{t-s}\, C_f (x-y),
$$
where $C_f$ is smooth and satisfies some additional assumptions detailed below. Our main result is the following:

\begin{thm}\label{thm:main_intro}
 Let $H \in (0,1)$ and let the forcing $f$ be 
\begin{equation}\label{eq:rforcing-phy-intro}
f(t,x)= \int_{\R_y^d} \varphi (x-y)\, dW(t,y),
\end{equation}
where $dW$ is a space-time Gaussian real white noise and $\varphi \in  \Sc (\R_x^d)$ is a radial function such that~$\widehat{\varphi}(k)=0$ for all $|k| < \kappa$.

\begin{enumerate}[(i)]
\item The transport equation in wavenumber space \eqref{eq:maineq-div-intro} with source \eqref{eq:rforcing-phy-intro} can be rigorously formulated in physical space as an a.s.\  well-posed PDE.
Moreover, at any $t>0$, the solution $u(t,x)$ has finite variance and a.s.\ smooth paths with respect to $x$. 
\item As $t\rightarrow\infty$, $u(t,x)$ converges in law to a zero-mean Gaussian field $u_{\infty}(x)$ which has a.s.\ $\alpha$-H\"older continuous paths for any $0<\alpha<H$. 
\item  The correlations are given by
\begin{equation*}\label{eq:correlations_infty-intro}
\E [u_{\infty} (x_1) u_{\infty}(x_2)] = C(d,H) \,\mathcal{K}_H(x_1-x_2) - (\mathcal{K}_H \ast \mathcal{J}_H) (x_1-x_2),
\end{equation*}
where 
$$
\mathcal{K}_H := \mathcal{F}^{-1} \left[ \chi_{|k| > \kappa} |k|^{-(2H+d)} \right] ,
$$
while~$C(d,H)$ is an explicit constant and the function $\mathcal{J}_H \in \Sc(\mathbb{R}_x^d)$ depends explicitly on $\varphi$ in~\eqref{eq:rforcing-phy-intro}. 
\end{enumerate}
\end{thm}

A more detailed version of this result is presented in \Cref{thm:main_H_pos}, page~\pageref{thm:main_H_pos}.

\begin{rk}\label{rk-kappa-source}
The parameter $\kappa$ can be chosen as the smallest non-vanishing wavenumber in the support of the Fourier transform of the forcing. Following the analogy with the Navier-Stokes equations presented in the introduction, $\kappa$ may be interpreted as a quantity linked to the inverse of the \emph{integral lengthscale}\footnote{
If  $L_f$ is the integral lengthscale, then there exists two real positive numbers $a<b$ such that the support of $\widehat{f}$ is contained in the annulus of inner radius $\frac{a}{L_f}$ and outer radius $\frac{b}{L_f}$. Thus one may set $\kappa = \frac{a}{L_f}.$},
 i.e the typical lengthscale of the correlations of the forcing. The fact that our force acts at large but finite scales means that $\kappa$ is small but non-zero.
The kernel ${\mathcal K}_H$ is a function when $H \in (0,1)$ and $\kappa >0$. However, when $\kappa=0$ we  have the following operator: 
$$
{\mathcal K}_H \underset{\kappa \to 0}{\longrightarrow} (-\Delta)^{-\left(H+\frac{d}{2}\right)}.
$$
\end{rk}

\begin{rk} Note that the limiting Gaussian field $u_{\infty}$ shares some statistical properties, such as roughness, with statistical homogeneous {\it{fractional gaussian fields}} \cite{cohen2013fractional, LodhiaFGF} defined by $$(-\Delta)^{-\frac{H+\frac{d}{2}}{2}} dW,$$ that are classically encountered in the turbulence literature \cite{kraichnan1968small,chaves2003lagrangian,chevillard2019skewed}. Indeed, for $H \in (0,1)$ both have a.s.\ $\alpha$-H\"older continuous paths for any $0<\alpha<H$ and one can show that
\begin{equation}\label{FGF-L-model}
u_{\infty} \underset{\text{(in law)}}{=} C(d,H)  \mathcal{F}^{-1} \left[ \chi_{|k| > \kappa} \right] \ast (-\Delta)^{-\frac{H+\frac{d}{2}}{2}} dW - u_{\mathrm{reg}}.
\end{equation}
Here, $u_{\mathrm{reg}}$ is a smooth zero-mean Gaussian field with correlations
$$
\E [u_{{\mathrm{reg}}} (x_1) u_{{\mathrm{reg}}}(x_2)] =  (\mathcal{K}_H \ast \mathcal{J}_H) (x_1-x_2),
$$
which is a smooth function with respect to $(x_1-x_2)$ even if $H \leq 0$.
The case $H \in [-d/2,0]$ will be discussed in Section \ref{sec-rigorous}.
\end{rk}

\begin{rk}\label{rk:negative_H}
When $H \in [-d/2,0]$, as $t\rightarrow\infty$, $u(t)$ still converges to a zero-mean Gaussian field $u_{\infty}$, but this field is not necessarily H\"older continuous with respect to $x$ anymore. In this case, one may view $u_{\infty}$ as a distribution living in the dual of an appropriate test function space $\mathcal{T}$, see \Cref{thm:main_H_neg} for more details. The correlation structure of the limiting Gaussian measure is given by:
\begin{equation}\label{corr-infty-intro}
 \E [  \langle u_{\infty} , g_1 \rangle {\langle u_{\infty}, g_2 \rangle}] = 
 \int_{\R_k^d}  \chi_{|k| >  \kappa} \,|k|^{-(2H+d)} \,  \left(  C(d,H) - {\widehat{\mathcal{J}}}_H(k) \right)\,\widehat{g_1}(k) \,\overline{\widehat{g_2}(k)}  \, dk 
\end{equation}
for any test functions $g_1,g_2 \in \mathcal{T}$, where $\langle \cdot ,\cdot \rangle $ stands for the duality product in $\mathcal{T}$. 
In \Cref{thm:main_H_neg}, we show that, for any test functions in $\mathcal{T} \cap \Sc(\mathbb{R}_d^x)$, \eqref{corr-infty-intro} yields a rate of convergence proportional to $(ct)^{-(2H+d+2n)}$ for $n$ as large as desired. The expression \eqref{corr-infty-intro}  corresponds to the energy-spectrum picture described above. Indeed, $|k|^{-(2H+d)}$ corresponds precisely to the inertial range previously described, while ${\widehat{\mathcal{J}}}_H(k)$ captures the contribution from the forcing, which  is a correction in the injection range. 
In fact, if the source is spectrally supported in small wavenumbers, then ${\widehat{\mathcal{J}}}_H(k)$ vanishes in the inertial range.
\end{rk}

Finally, let us highlight the difference between the properties of the solution at finite and infinite time. At finite time, the solution is smooth with respect to $x$, whereas at infinite time the solution is only H\"older continuous (or even rougher if $H\leq 0$, as explained in \Cref{rk:negative_H}). This loss of regularity at infinite time is what is expected in linear turbulence. Turbulence is usually associated to a nonlinear equation. For example, in the case of wave turbulence, nonlinearities create wave interactions which allow the transfer of energy to higher and higher wavenumbers. Such transfers of energy typically result in a loss of regularity.
However, nonlinearities might not be the only way in which such loss of regularity can occur. Indeed, Y. Colin de Verdi\`ere and L. Saint-Raymond~\cite{CSR} have shown that, in the context of internal waves, a loss of regularity can also take place in the case of a linear equation with an operator of degree 0.
Linear equations with operators of degree 0 are also common whenever one introduces a dispersive perturbation in a hyperbolic system. In such cases, these operators of degree 0 are used to model wave propagation under strong dispersive effects and they are responsible for memory effects.
For example, in the context of wave-energies, the second author and D. Lannes show that the waves generated by a moving floating object are governed in the linear regime by a non-local transport equation of degree 0, see \cite{Beck-Lannes}. In the context of electrical circuits, there are cases in which 1D models of electromagnetic waves propagating along a coaxial cable are governed by operators of degree 0, see for instance \cite[Chapter 5]{beck:tel-01421740}. \\

One issue of our model is that it only features a single H\"older exponent. The velocity field of a concrete turbulent fluid consists of many H\"older exponents, i.e. the H\"older-regularity of the velocity field $u(t,x)$ around a point $x \in \mathbb{R}^d$ depends on the point itself. This is known as the multifractal formalism \cite{chevillard:tel-01212057,cohen2013fractional}. The term multifractal refers to the fact that the sets of points with same regularity are often fractal. Moreover, our model does not capture finer descriptions (beyond the variance) of the distribution of the increments of the velocity field. Such descriptions should quantify its non-Gaussian and intermittent nature \cite{chevillard2019skewed}, and therefore our linear model does not suffice. 
It is known that one can construct a multifractal and intermittent field with the theory of Gaussian multiplicative chaos \cite{Rhodes:2013iua}, however our actual goal is to obtain a multifractal and intermittent field dynamically, i.e. as the solution to a non-linear equation forced by a white-noise in time, but that admits a rigorous mathematical treatment.
We can also consider a forcing which is not a white-noise in time whose temporal correlation function is given by an oscillating function in order to make a comparison with \cite{CSR, Zworski,AlmoNigam} and \cite{Carles-quasi-rec}. Finally, we could investigate other linear models of cascade such as in the case of a compact operator plus a potential of degree 0 as in \cite{Maspero}. These issues will be tackled in future papers.


\subsection{Outline} 

The article is organized as follows. In \Cref{sec-1D}, we present a simple transport equation that converges to a complex white noise (up to lower order terms). A small tweak to this model allows us to construct a model that gives rise to a \emph{real} white noise. In \Cref{sec-heuristic}, we explain how to generalize the latter model to higher dimensions and give a heuristic proof of the main results in this paper. In \Cref{sec-rigorous}, we provide a mathematically rigorous study of our model: we introduce the right functional setting, we develop a global well-posedness theory and give a complete description of the asymptotic behavior of the solution, as well as its properties. This constitutes the proof of \Cref{thm:main_intro}. Finally, in \Cref{sec-simulations} we propose a numerical method and conduct numerical simulations in dimensions 1, 2 and 3 to validate our model.

\subsection{Acknowledgements}

We thank Oliver B\"{u}hler for his interesting suggestion about replacing the pure transport term $\pa_k$ in \eqref{eq:1d_cmain_fourier} by $\text{sign}(k)\pa_k$ as a way of fixing the physically undesirable behavior of the solution. 
We thank Charles-Edouard Br\'ehier and the referee for their careful reading and their corrections of a previous version of this manuscript.
All five authors were funded by the Simons Collaboration Grant on Wave Turbulence, Simons Award ID: 651475 and 651675.

\subsection{Notation}

For an integrable function $f:\R^d \rightarrow \C$, we denote by $\widehat{f}$ its {\bf{Fourier transform}}, namely
\[
\forall k \in  \R_k^d, \qquad  \widehat{f}(k) := \Fc f(k) :=\int_{\R_x^d} e^{-2\pi i x\cdot k} \, f(x)\, dx.
\]
Whenever defined, the inverse Fourier transform is
\[
\forall x \in  \R_x^d, \qquad  f(x) = \Fc^{-1} \widehat{f} (x) = \int_{\R_k^d} e^{2\pi i x\cdot k} \, \widehat{f}(k)\, dk.
\]
It is well know that the Fourier transform is an isometry from $L^2(\R_x^d)$ to $L^2(\R_k^d)$, from $\Sc(\R_x^d)$ to $\Sc(\R_k^d)$, where $\Sc(\R_x^d)$ denote the space of {\bf{Schwartz functions}} (i.e. smooth functions whose derivatives are rapidly decreasing), and from $\Sc'(\R_x^d)$ to $\Sc'(\R_k^d)$ where $\Sc'(\R_x^d)$ denote the space of tempered distribution (i.e. the dual space of $\Sc(\R_x^d)$). We will denote by $\langle \cdot , \cdot \rangle$ the duality product between $\Sc'$ and $\Sc$.

We will also need some spaces that quantify the regularity of functions more precisely. For a fixed integer $n$,  the {\bf{Sobolev space}} of order~$n$ is defined by
$$
H^n(\mathbb{R}^d) := \{ u \in L^2(\R^d) \, | \, \partial_{x_i}^j u \in L^2(\R^d) \, \text{ with } 1 \leq i \leq d \, \text{ and }  0 \leq j \leq n \},
$$
 and their dual spaces are denoted by $H^{-n}(\mathbb{R}^d)$. We will denote by $\langle \cdot , \cdot \rangle_{H^{-n},H^n}$ the duality product between $H^n$ and $H^{-n}$. For $\alpha\in (0,1)$ the {\bf{H\"older space}} $C^{0,\alpha}(\mathbb{R}^d)$ is defined by
$$
C^{0,\alpha}(\mathbb{R}^d):=  \{ u \text{ continuous and bounded } | \, \exists C >0, \,  \forall x, \ell \in \mathbb{R}^d, | \ell | \leq 1,  | \delta_\ell u (x) | \leq C |\ell|^{\alpha} \},
$$
where $\delta_\ell$ denotes the {\bf{increment}} defined by
\[
\delta_\ell u(x):= u(x+ \ell)- u(x)
\quad \text{for } x, \ell \in \mathbb{R}^d.
\]

We will denote by $\chi_A$ the {\bf{characteristic function}}\footnote{This means that $\chi_A(k)=1$ if $k\in A$ and $\chi_A(k)=0$ if $k\notin A$.} of the set $A$. \\

Let $(\Omega, \sigma(\Omega), \mathbb{P})$ be a probability space.
 A {\bf{Gaussian field}} $u: \mathbb{R}^d \to L^2(\Omega)$ is a field such that for all $n \geq 1$ and for all $(x_1, x_2, \dots, x_n)\in (\R^d)^n$, the random vector $(u(x_1), u(x_2), \dots, u(x_n))$ is a Gaussian random vector. When $d=1$, a 1D Gaussian field is usually called a Gaussian process. A {\bf{Gaussian random measure}} $\mu$ acting on $\Sc(\mathbb{R}^d)$ is a random tempered distribution such that for every $g \in \Sc(\mathbb{R}^d)$, the random variable $\langle \mu, g \rangle$ is a centered Gaussian. A {\bf{white noise}} $dW(x)$ is a Gaussian random measure acting on $L^2(\mathbb{R}^d)$ that satisfies the following for any functions $f,g \in L^2(\mathbb{R}^d)$
\[
 \mathbb{E} \left[ \left( \int_{\mathbb{R}^d} f(x) dW(x) \right) \left( \int_{\mathbb{R}^d} g(x) dW(x) \right) \right] =  \int_{\mathbb{R}^d} f(x) g(x) dx.
\]
Since we always integrate a deterministic function against $dW(x)$, the choice between It\^{o} and Stratonovich integrals is unimportant, since Wiener integration suffices (see, for instance, \cite{Janson}).

\section{One-dimensional transport in wavenumber space}\label{sec-1D}

\subsection{Building one-dimensional white noise: real vs complex}

In order to mimic the transport of energy from large scales to small scales, the authors in \cite{ACM,AC} proposed a simple transport equation in Fourier space. To present these ideas, we first consider a one-dimensional model for a velocity field $u(t,x)$, whose spatial Fourier transform aims to solve the linear evolution
\begin{equation}\label{eq:1d_cmain_fourier}
\begin{cases}
\pa_t \widehat{u} (t,k) + c\,\pa_k \widehat{u}(t,k)= \widehat{f}(t,k), \quad (t,k)\in (0,\infty) \times \R,  \\
\widehat{u}(t,k)|_{t=0}=0.
\end{cases}
\end{equation}
Here $c>0$ is fixed and can be viewed as a transport rate in wavenumber space. On the right-hand side of \eqref{eq:1d_cmain_fourier}, we have included an additive term $ \widehat{f}$ which is the Fourier transform of a spatial forcing term. The support of $\widehat{f}$ is localized at small wavenumbers, which is consistent with the assumption that the forcing term in physical space acts at large scales. 

As we can see, the dynamical evolution proposed in \eqref{eq:1d_cmain_fourier} is a genuine transport equation, and only the presence of a forcing makes it inhomogeneous. Adopting such a setup immediately imposes the complex nature of the velocity field in physical space, as it can be seen when formally taking the inverse Fourier transform of \eqref{eq:1d_cmain_fourier}, and obtaining the following evolution in physical space:
\begin{equation}\label{eq:1d_cmain}
\begin{cases}
\pa_t u(t,x) - 2\pi icx\, u(t,x) = f(t,x), \quad (t,x)\in (0,\infty) \times \R,  \\
u(t,x)|_{t=0}=0.
\end{cases}
\end{equation}
Note that the operator in \eqref{eq:1d_cmain} corresponds to multiplication by the space variable $2\pi icx$. In \cite{ACM}, it is shown that when the forcing $f$ is a white noise in time and statistically homogeneous in space, then the solution to \eqref{eq:1d_cmain}, $u:(0,\infty)_t\times\R_x \longrightarrow \C$, converges\footnote{Up to lower order terms, see \Cref{thm:1d_cmain} for a full asymptotic expansion.
} 
to a \emph{complex} white noise in space as $t\rightarrow\infty$. In other words, the evolution that has been proposed, expressed in Fourier space as a transport equation \eqref{eq:1d_cmain_fourier} and in physical space as an equation \eqref{eq:1d_cmain} involving an operator of degree 0 (multiplication by $-2\pi icx$), is able to transfer energy through scales. Moreover, the solution is statistically homogeneous at any time and it develops the regularity of a white noise as time goes on (technically it's a sudden drop in regularity at $t=\infty$). To complete the program suggested by the phenomenology of turbulence, the additional linear action of a fractional operator allows, in a similar setup, to generate a solution with asymptotic H\"{o}lder-type regularity of parameter $H\in(0,1)$ instead of the one of the white noise, as explained in \cite{ACM}.

If some energy is introduced by the forcing at a negative wavelength $k<0$, the transport equation \eqref{eq:1d_cmain_fourier} will move it to smaller negative wavenumbers, going through $k=0$, and then to infinitely large positive wavenumbers $k\to\infty$. In order to avoid this pathological behavior, one could replace $\pa_k$ by $\pa_{|k|}=\text{sign}(k)\pa_k$ which leads to a transport in the direction of $|k|$ instead of $k$. Note however that $\text{sign}(k)\pa_k$ is not properly defined at $k=0$, and so one needs to be careful in order to propose a well-posed mathematical problem. With this in mind, the heart of this article will be the theoretical and numerical study of the following formal evolution:
\begin{equation}\label{eq:1d_rmain_fourier}
\begin{cases}
\pa_t \widehat{u} (t,k) + c\,\pa_{|k|} \widehat{u}(t,k)= \widehat{f}(t,k), \quad (t,k)\in (0,\infty) \times \R,  \\
\widehat{u}(t,k)|_{t=0}=0, \\
\widehat{u}(t,k)|_{|k|=0}=0.
\end{cases}
\end{equation}
Note that it is necessary to introduce a transmission condition between negative and positive $k$. In~\eqref{eq:1d_rmain_fourier}, we have decided to add the boundary condition $\widehat{u}(t,k)|_{|k|=0}=0$ to decouple negative from positive wavenumbers, so that no energy crosses $k=0$. In particular, this means that the integral over space of $u$ is zero for all times. This new dynamics proposed in \eqref{eq:1d_rmain_fourier} can be written in physical space after formally applying the inverse Fourier transform:
\begin{equation}\label{eq:1d_rmain}
\begin{cases}
\pa_t u + 2\pi c\, x \mathcal{H} (u) = f, \quad (t,x)\in (0,\infty) \times \R,  \\
u|_{t=0}=0, \\
\displaystyle \int_{\R_x} u \, dx=0,
\end{cases}
\end{equation}
where $\mathcal{H}$ denotes the Hilbert transform defined in the usual way:
\begin{equation}
\mathcal{H}f(x):= \frac{1}{\pi}\mbox{p.v.} \int_{\R_y} \frac{f(y)}{x-y}\, dy =-\frac{1}{\pi} \lim_{\epsilon \rightarrow 0^{+}} \int_{\epsilon}^{\infty} \frac{f(x+y)-f(x-y)}{y} \, dy.
\end{equation}
Notice that we have used the fact the integral of $u(t,x)$ over space vanishes to get the expression of~\eqref{eq:1d_rmain}. Notice also that, despite the fact that the spectral evolutions \eqref{eq:1d_cmain_fourier} and \eqref{eq:1d_rmain_fourier} (whose equivalent expressions in physical space are provided respectively in  \eqref{eq:1d_cmain} and \eqref{eq:1d_rmain}), look very similar, the solution to the new dynamics  \eqref{eq:1d_rmain} is now real-valued. Equivalently, the dynamics in Fourier space \eqref{eq:1d_rmain_fourier} conserves the Hermitian symmetry of an appropriate initial condition, here assumed to be zero.  Moreover, the solution $u:(0,\infty)_t \times \R_x \rightarrow \R$  of \eqref{eq:1d_rmain} can be shown to asymptotically converge to a \emph{real} white noise in space. As we will explain in the sequel, the additional linear action of a fractional operator will allow the generation, from smooth forcing, of a real fractional Gaussian field. 


Finally, it is tempting to generalize \eqref{eq:1d_rmain_fourier} to higher dimensions by replacing $\partial_{|k|}$ by $\frac{k}{|k|}\cdot \nabla_k$. As we will develop in \Cref{sec-rigorous}, this eventually generates a statistically homogeneous and isotropic solution that will converge to a \emph{real} $d$-dimensional Gaussian random measure which is rougher than a white noise (in space) whenever $d>1$. As we will explain in the sequel, the additional linear action of a fractional operator will allow us to generate a \emph{real} $d$-dimensional Gaussian random measure with the desired H\"older regularity.  Interestingly, it is not obvious to generalize \eqref{eq:1d_cmain_fourier} to space dimension $d\ge 1$ which would generate a similar statistically homogeneous and isotropic solution in physical space. We provide at the end of the section some additional discussions on this matter.

%

However, for the time being we focus on developing a good understanding in the one-dimensional setting. In the case of \eqref{eq:1d_rmain}, we have the following result:

\begin{thm}\label{thm:1d_rmain}
Let the forcing $f$ in \eqref{eq:1d_rmain} be 
\begin{equation}\label{eq:rforcing}
f(t,x)= \int_{\R_y} \varphi (x-y)\, dW(t,y),
\end{equation}
where $dW$ is a space-time Gaussian real white noise, and $\varphi\in\Sc (\R)$ is a non-negative, non-identically null, even function with null average. Then:
\begin{enumerate}[(i)]
\item Equation \eqref{eq:1d_rmain} admits a global (in time) solution $u(t,x)$, which is a Gaussian process with a.s.\ smooth paths in $x$, and with $\alpha$-H\"older continuous paths in $t$ for any $0<\alpha<1/2$. 
\item As $t\rightarrow\infty$, the solution $u(t)$ converges in $\Sc' (\R)$ to a random Gaussian measure $u_{\infty}$ acting on $\Sc (\R)$ with zero-mean, i.e. for every $g\in\Sc(\R)$, $\langle u(t),g\rangle$ converges in law to a complex Gaussian random variable $\langle u_{\infty},g\rangle$ with $\E [ \langle u_{\infty} , g \rangle] = 0$.
\item We have the following asymptotic behavior:
\begin{equation}\label{eq:corr_infty_1d_r}
\begin{split}
\E [ \langle u_{\infty}, g_1 \rangle {\langle u_{\infty}, g_2 \rangle}] =  & \ \lim_{t\rightarrow\infty} \E [ \langle u(t), g_1 \rangle {\langle u(t), g_2 \rangle}] \\
=& \ 
C
 \int_{\R_x} g_1(x) \,g_2(x)\, dx  - \int_{\R_x\times\R_y}  \mathcal{I} (x-y)\, g_1(x) \, g_2(y)\, dx\, dy.
\end{split}
\end{equation}
for any $g_1,g_2\in \Sc (\R)$. Here $C>0$ is a constant and $\mathcal{I}$ is an explicit continuous, even function that depends on $\varphi$ in \eqref{eq:rforcing}.
\end{enumerate}
\end{thm}

A more detailed version of this result is presented in \Cref{thm:main_H_neg}, see also \Cref{sec-whitenoise}. Note that the first term on the right-hand side of \eqref{eq:corr_infty_1d_r} corresponds to a delta function, while the second term given by $\mathcal{I}$ is a smooth lower order term.

\begin{rk}\label{rk:rforcing-correlated}
The forcing introduced in \eqref{eq:rforcing} is indeed a Gaussian white noise in time and statistically homogeneous in space, i.e.
\begin{equation}\label{eq:rforcing-correlated}
\E [f(s,x) f(t,y)] = C_f ( x-y)\, \delta_{s-t} \, ,
\end{equation}
where the spatial correlation function $C_f = \varphi\ast\varphi$ is a convolution. Given that $\varphi\in\Sc (\R)$, $C_f\in\Sc (\R)$. The constant $C$ in \eqref{eq:corr_infty_1d_r} is precisely
\[
C=\frac{C_f(0)}{2c} =\frac{1}{2c} \, \int_{\R_x} |\varphi (x)|^2 \, dx >0 .
\]
\end{rk}

\begin{rk}\label{rk:solution-loss-reg}
 The solution to the stochastic PDE \eqref{eq:1d_rmain} is an explicit Gaussian It\^{o} process. A precise formula will be given in \Cref{sec-rigorous}. Even if this solution is continuous in time and space, one can lose regularity at $t=+\infty$, which is why one needs to consider  $u_{\infty}$ on the left-hand side of \eqref{eq:corr_infty_1d_r} as a distribution.
\end{rk}

The proof of this theorem is posponed to the next section where a more general case, i.e. multidimesional white noise, will be tackled. Before we develop the techniques needed to prove \Cref{thm:1d_rmain}, it is important to understand the asymptotic behavior of solutions to \eqref{eq:1d_cmain} in the complex setting, which is less technical, but informative.

In this setting, we have the following result:

\begin{thm}\label{thm:1d_cmain}
Let the forcing $f$ in \eqref{eq:1d_cmain} be 
\begin{equation}\label{eq:cforcing}
f(t,x)= \int_{\R_y} \varphi (x-y)\, dW(t,y),
\end{equation}
where $dW$ is a space-time Gaussian complex white noise, and $\varphi\in\Sc (\R)$ is a complex, non-identically null, even function. Then:
\begin{enumerate}[(i)]
\item Equation \eqref{eq:1d_cmain} admits a global (in time) solution $u(t,x)$, which is a Gaussian process with a.s.\ continuous paths in time and space.
\item As $t\rightarrow\infty$, the solution $u(t)$ converges (in $\Sc'(\mathbb{R}^d)$) to a random Gaussian measure $u_{\infty}$ acting on $\Sc(\mathbb{R}^d)$, i.e for every $g\in\Sc(\R)$, $\langle u(t),g\rangle$ converges in law to $\langle u_{\infty},g\rangle$ where $u_{\infty}$ is a random Gaussian measure with $\E [ \langle u_{\infty} , g \rangle] = 0$.
\item We have the following asymptotic behavior (in the sense of distributions).\\
 For any $g_1,g_2\in \mathcal{S}(\R)$,
\begin{equation}\label{eq:casymp}
\begin{split}
 \E [ \langle u_{\infty}, g_1 \rangle \overline{\langle u_{\infty}, g_2 \rangle}]
&=\lim_{t\rightarrow\infty} \E [ \langle u(t), g_1 \rangle \overline{\langle u(t), g_2 \rangle}] 
\\
&= \frac{1}{2 c} \, C_f(0)\, \int_{\mathbb{R}_z} g_1(z) g_2(z) dz\\
& +  \frac{1}{2 \pi i c}  \,\mbox{p.v.}\ \int_{\mathbb{R}_z} \frac{C_f (z)}{z}  \left( \int_{\mathbb{R}_y} g_1(z + y) g_2(y) dy \right) dz\ .
\end{split}
\end{equation}
The function $C_f=\varphi\ast \overline{\varphi}$ is the spatial correlation function given by
\begin{equation}\label{eq:cforcing-correlated}
\E [f(s,x) \overline{f(t,y)}] = \delta_{s-t} \, C_f ( x-y),
\end{equation}
and $\mbox{p.v.}\ \frac{C_f(z)}{z}$ is the principal value of the distribution $C_f(z)/z$. That is, for any test function $g\in \mathcal{S}(\R)$,
\begin{equation}\label{eq:def_vp}
\left\langle \mbox{p.v.} \frac{C_f(z)}{z}, g \right\rangle  := \mbox{p.v.} \ \int_{\R} \frac{C_f(z) g(z)}{z}\, dz = \int_0^{\infty} C_f(z)\,\frac{g(z)-g(-z)}{z}\, dz.
\end{equation}
\end{enumerate}
\end{thm}

As we mentioned in \Cref{rk:rforcing-correlated}, the forcing in \eqref{eq:cforcing} is a complex Gaussian white noise in time and statistically homogeneous in space. 
Moreover, the solution admits an explicit formula:
\begin{equation}\label{eq:solution-proecc-complex}
u(t,x)= \int_0^t  \int_{\R_y} e^{2 \pi i c x(t-s)}\, \varphi (x-y)\, dW(s,y).
\end{equation}
The same comments as in \Cref{rk:solution-loss-reg} apply in this case.

%

\begin{rk} The asymptotic expansion \eqref{eq:casymp} remains valid when testing against functions with a finite number of derivatives. Indeed, $u_{\infty}$ in \eqref{eq:casymp} can also be interpreted as a Gaussian random measure in $H^{-n}(\mathbb{R})$ for any integer $n \geq 2$. More precisely, we will show that for any test functions $g_1,g_2 \in H^{n}(\mathbb{R})$ one gets
\begin{equation}
\begin{split}
 \E [\langle u(t),g_1\rangle_{H^{-n},H^n} \overline{\langle u(t),g_2\rangle_{H^{-n},H^n}}] 
\underset{t \to \infty }{\sim} &\ \frac{C_f(0)}{2 c}\, \int_{\R} g_1(z) g_2 (z)\, dz \\
& +  \frac{1}{2 \pi i c} \, \mbox{p.v.}\ \int_{\mathbb{R}_z} \frac{C_f (z)}{z}  \left( \int_{\mathbb{R}_y} g_1(z + y) g_2(y) dy \right) dz \\
& + \mathrm{d}(n) \, \norm{g_1}_{H^n} \, \norm{g_2}_{H^n}\, \left( \frac{1}{ct} \right)^{n-1}.
\end{split}
\end{equation}
where $\mathrm{d}(n)$ depends only on $n$. This characterizes the \emph{rate of convergence}. 
\end{rk}

\begin{rk}\label{rk-regular correction}
One interpretation of \eqref{eq:corr_infty_1d_r} (resp. \eqref{eq:casymp}) is that the correlation function (resp. the real part of it) asymptotically behaves like a white noise.
However, this theorem gives a lower-order correction, in the sense that the regularity of the correction is higher than that of the white noise.
In \eqref{eq:corr_infty_1d_r}, such a regular correction is given by a Schwartz function, whereas in \eqref{eq:casymp}, the regular correction is a purely imaginary principal value which has no singularity at zero: by \eqref{eq:def_vp}, the principal value is ``controlled'' by $C_f'(0)$ near zero. It is important to note that in both cases the regular correction is fast-decaying.
\end{rk}

\begin{rk} 
Note that we recover the result in proposition 2.1 in \cite{ACM}, i.e.
\begin{equation}\label{eq:v_space}
\lim_{t\rightarrow\infty} \E [u(t,x)\overline{u(t,y)}] = \frac{1}{2c} \, C_f (0) \delta_{x-y}
\end{equation}
as long as one only tests against even functions with respect to the variable $x-y$, as is easily seen from the right-hand side of \eqref{eq:def_vp}.

One way to recover a result similar to that in \cite{ACM} that holds for all test functions is to define the function $v(t,x)=e^{- \pi i ct x} u(t,x)$. This function now satisfies
\[
\E [v(t,x)\overline{v(t,y)}] = t\, \sinc \left(ct(x-y)\right) \, C_f (x-y)
\]
where $\sinc(x):= \frac{ \sin(\pi x)}{ \pi x}$ denotes the normalized sinc function.
It immediately follows that
\begin{equation}\label{eq:v_space2}
\lim_{t\rightarrow\infty} \E [v(t,x)\overline{v(t,y)}] = \frac{C_f (0)}{c} \,  \delta_{x-y}.
\end{equation}
However, it is unclear whether this transformation of $u$ is an interesting object from the physical viewpoint. 
Assuming one can take the Fourier transform, \eqref{eq:v_space2} can be rewritten in wavenumber space as
\begin{equation}
\lim_{t\rightarrow\infty} \E [\widehat{v}(t,k)\overline{\widehat{v}(t,k')}] = \lim_{t\rightarrow\infty} \E \left[\widehat{u}\left(t,k+ \pi c t \right)\overline{\widehat{u} \left(t,k'+\pi c t\right)}\right] = \delta_{k-k'} \, C_f(0).
\end{equation}
The transformation given by $v$ is therefore equivalent to computing the correlation between the~$k+\pi c t$ and $k'+\pi c t$ Fourier modes as $t\rightarrow\infty$. 
\end{rk}

\subsection{Proof of \Cref{thm:1d_cmain}}
First of all, note that equation \eqref{eq:1d_cmain} admits the explicit solution~\eqref{eq:solution-proecc-complex} thanks to the Duhamel formula. \\

{\bf{Step 1:}} The solution $u$ is a well defined Gaussian field whose limit at $t \to \infty$ is a Gaussian random measure.

Clearly, $u(t,x)$ is a well defined It\^o process with zero average and variance
\begin{equation}\label{notL^p}
\E [ |u(t,x)|^2 ] = \int_0^t  \int_{\R_y} \, |\varphi (x-y) |^2 \, dy ds = t \norm{\varphi}_{L^2}^2.
\end{equation}
For any test function $g$ (we will soon see that actually $g\in H^1(\mathbb{R})$ suffices), one gets
\[
\langle u(t),g \rangle =  \int_0^t    \int_{\R_y} \left( \int_{\R_x} e^{2 \pi i c x(t-s)}\, \varphi (x-y)\, g(x) \, dx \right)\, dW(s,y).
\]

For fixed $t$, this random variable has the same distribution as
 \begin{equation}\label{eq:v}
\langle v(t),g \rangle :=  \int_0^t    \int_{\R_y} \left( \int_{\R_x} e^{-2 \pi i c x s}\, \varphi (x-y)\, g(x) \, dx \right)\, d W(s,y).
\end{equation}
As $t \to \infty$, one has the mean-square convergence property
\begin{equation}\label{eq:u_infty}
\lim_{t \to \infty} \langle v(t),g \rangle = \langle v_{\infty} ,g \rangle=  \int_0^\infty   G(s,y) \, d W(s,y)
\end{equation}
where
$$
G(s,y) = \int_{\R_x} e^{-2 \pi i c x s}\, \varphi (x-y)\, g(x) \, dx.
$$
In order to justify the convergence result above, it is sufficient to assume that $G \in L^2( \R_s^+ \times \R_y)$, which we set out to prove next. Firstly, the Young convolution inequality immediately yields
\[
\norm{G}_{L^2([0,1]_s \times \R_y)} \leq \norm{\varphi}_{L^1} \norm{g}_{L^2}.
\]
Next, to handle the case $|s|\geq 1$, we assume that $g\in H^1(\R)$ and we integrate by parts:
$$
2 \pi i c s\, G(s,y) = - \int_{\R_x} e^{-2 \pi i c x s}\, \partial_x \big( \varphi (x-y)\, g(x) \big) \, dx.
$$
By the Young convolution inequality,
$$
\norm{2 \pi i c  s G(s,y)}_{L^2_y}  \leq  \norm{g'}_{L^2} \norm{ \varphi}_{L^1} + \norm{ g}_{L^2} \norm{ \varphi'}_{L^1} .
$$
Thus one easily finds that:
$$
\norm{ G(s,y)}_{L^2( \R_s^+ \times \R_y)} \lesssim \frac{1}{c} \norm{g}_{H^1} \norm{ \varphi}_{W^{1,1}},
$$
which justifies \eqref{eq:u_infty}. 

Since the mean and the variance of $\{ \langle v(t),g\rangle\}_{t>0}$ converge and since they coincide with those of $\{ \langle u(t),g\rangle \}_{t>0}$ for each fixed $t$, we deduce that $\langle u(t),g\rangle$ converges in law to a Gaussian random variable
\begin{equation}\label{eq:u_infty_law}
\langle u_{\infty} ,g \rangle \underset{\mathrm{law}}{=}  \int_0^\infty   G(s,y) \, d W(s,y).
\end{equation}
Moreover,  all the moments of $\{ \langle u(t),g\rangle \}_{t>0}$ converge to those of $\langle u_{\infty} ,g \rangle$.

{\bf Step 2:} 
We would like to show that 
\begin{equation}\label{eq:u_convergence}
  \E [ \langle u_{\infty}, g_1 \rangle \overline{\langle u_{\infty}, g_2 \rangle}]
  =\lim_{t\rightarrow\infty} \E [ \langle u(t), g_1 \rangle \overline{\langle u(t), g_2 \rangle}] .
\end{equation}
By  \eqref{eq:u_infty} and \eqref{eq:u_infty_law},
\[
  \E [ \langle u_{\infty}, g_1 \rangle \overline{\langle u_{\infty}, g_2 \rangle}] =   \E [ \langle v_{\infty}, g_1 \rangle \overline{\langle v_{\infty}, g_2 \rangle}] 
\]
for $v$ defined in \eqref{eq:v}. 

Since $\E [ \langle v(t), g_1 \rangle \overline{\langle v(t), g_2 \rangle}]= \E [ \langle u(t), g_1 \rangle \overline{\langle u(t), g_2 \rangle}]$ for each fixed $t$, our objective \eqref{eq:u_convergence} follows from the identity 
\begin{equation}\label{eq:v_convergence}
\E [ \langle v_{\infty}, g_1 \rangle \overline{\langle v_{\infty}, g_2 \rangle}] = \lim_{t\rightarrow\infty} \E [ \langle v(t), g_1 \rangle \overline{\langle v(t), g_2 \rangle}].
\end{equation}

In order to justify \eqref{eq:v_convergence}, we may simply use the mean-square convergence of $\langle v(t),g\rangle$ to $\langle v_{\infty}, g\rangle$ for any $g\in\Sc (\R^d)$, together with the following ``polarization'' identity:
\[
\begin{split}
\langle v(t), g_1 \rangle \overline{\langle v(t), g_2 \rangle} =&\  \frac{1}{4} \,\left( |\langle v(t), g_1+g_2\rangle|^2 - |\langle v(t), g_1-g_2\rangle|^2 \right) \\
& + \frac{i}{4} \, \left( |\langle v(t), g_1+ig_2\rangle|^2 - |\langle v(t), g_1-ig_2\rangle|^2\right).
\end{split}
\]



{\bf Step 3:} Calculation of the correlations as $t \to \infty$.

We start by computing the correlations for a finite time $t>0$.
\[
\E [ u(t,x_1) \overline{u(t,x_2)}] = \int_0^t e^{2 \pi i c (x_1-x_2)s} C_f (x_1-x_2)\, ds = \frac{e^{2 \pi i c t(x_1-x_2)}-1}{2 \pi i c(x_1-x_2)} \, C_f (x_1-x_2).
\]
This is a well-defined function for each finite $t$, but we must treat it as a distribution if we want to take the limit $t\to\infty$. To do so, we test it against some $g_1,g_2 \in \Sc(\R)$:
\begin{equation}\label{eq:def_psi}
\begin{split}
\E [ \langle u(t),g_1 \rangle  \overline{\langle u(t),g_2 \rangle } ] & =
 \int_{\mathbb{R}_{x_1} \times \mathbb{R}_{x_2}}  \E [ u(t,x_1) \overline{u(t,x_2)}] g_1(x_1) g_2(x_2) \, dx_1 dx_2 \\
& = \int_{\mathbb{R}_z}  \frac{e^{2 \pi i c tz}-1}{2 \pi i cz} \, \psi (z) \, dz 
\end{split}
\end{equation}
with $\psi(z)= C_f(z) \int_{\mathbb{R}_{y}} g_1(z+y)g_2(y)dy  $, $z=x_1-x_2$ and $y=x_2$.
Our goal is to study the last integral as $t \to \infty$. 

We start with a simple identity that gives us a way to integrate the function $(e^{ 2 \pi i c tz}-1)/(2 \pi i c z)$. Note that this function is not absolutely integrable in $\R$. However, its integral does converge conditionally, i.e. the final result might depend on how we integrate it. More precisely, we recall that for all $t >0$
\begin{equation}\label{eq:identity}
\int_{\R} \frac{e^{itz}-1}{iz}\,dz := \lim_{R\rightarrow\infty} \int_{-R}^{R} \frac{e^{itz}-1}{iz}\,dz = \pi. 
\end{equation}

As a result of \eqref{eq:identity}, we have that
\[
\lim_{R\rightarrow\infty} \int_{-R}^{R} \frac{e^{- 2 \pi i c tz}-1}{2 \pi i cz} \psi (z)  dz - \frac{\psi (0)}{2c} = \lim_{R\rightarrow\infty} \int_{-R}^{R} \frac{e^{- 2 \pi i c tz}-1}{2 \pi i cz} [\psi (z)  - \psi (0)] \, dz .
\]
In view of 
$$
\psi(0) = C_f(0) \int_{\mathbb{R}_y} g_1(y) g_2(y) dy ,
$$
it suffices to prove the following in order to obtain \eqref{eq:casymp}:
\begin{equation}\label{eq:main_goal}
\begin{split}
\lim_{t\rightarrow\infty}   \lim_{R\rightarrow\infty} \int_{-R}^{R} \frac{e^{- 2 \pi i c tz}-1}{2 \pi i cz} [\psi (z)  - \psi (0)] \, dz dy &=
 \frac{1}{2 \pi c} \int_0^\infty \frac{\psi(z)-\psi(-z)}{i z}\,dz 
 \\ &= \frac{1}{2 \pi i c} \, \mbox{p.v.} \int_{ \mathbb{R}_z} \frac{\psi(z)}{z} dz.
 \end{split}
\end{equation}
To prove \eqref{eq:main_goal}, we rewrite the left-hand side as follows: 
\begin{multline}\label{eq:towards_delta}
\lim_{R\rightarrow\infty} \int_{-R}^R \frac{e^{-2 \pi i ct z}-1}{2 \pi c} \frac{\psi(z)-\psi(0)}{iz}\,dz = \frac{1}{2 \pi i c} \Big( \lim_{R\rightarrow\infty}\int_{-R}^R e^{-2 \pi i ct z} \, \frac{\psi(z)-\psi(0)}{z}\,dz \\
- \lim_{R\rightarrow\infty}\int_{-R}^R \frac{\psi(z)-\psi(0)}{z}\,dz \Big)
\end{multline}
Note that the last term gives the desired limit after using the fact that
\[
 \int_{-R}^R \frac{\psi(z)-\psi(0)}{z}\,dz =  \int_0^R \frac{\psi(z)-\psi(-z)}{z}\,dz
\]
and taking $R\to\infty$. 

The final step is to show that the first term on the right-hand side of \eqref{eq:towards_delta} tends to zero as~$t\rightarrow\infty$. Note that 
\[  F(z):=\frac{\psi(z)-\psi(0)}{z}\]
 is not integrable in $\R_z$. However the following lemma shows that its derivatives have better properties. Its proof  is postponed to the end of the proof of \Cref{thm:1d_cmain}. Recall that $\psi$ was defined in terms of $g_1,g_2$ right after \eqref{eq:def_psi}.
 
\begin{lem}\label{CV-t^-n-lemma-technic}
Let $n \geq 1$. If $g_1\in H^{n+1}(\mathbb{R}) $ and $ g_2 \in L^2(\R)$ then $F \in W^{n,1}(\mathbb{R}_{z})$ and 
 \begin{align}
\lim_{|z| \to \infty} (\pa_z^{m} F) (z) & =0,  \text{ for any $0 \leq m\leq n$},
 \label{eq:boundary} \\
 \norm{\partial_z^m F}_{L^1(\mathbb{R}_z)} & \lesssim \norm{g_1}_{H^{m+1}} \, \norm{g_2}_{L^2}  \text{ for any $1 \leq m\leq n$}.  \label{eq:bilinear}
 \end{align}
\end{lem}

We integrate by parts the first term on the right-hand side of \eqref{eq:towards_delta}
 \[
 \begin{split}
\int_{-R}^R e^{-ictz} \, \frac{\psi(z)-\psi(0)}{z}\,dz & = -\int_{-R}^R \frac{1}{ict} \, \pa_z e^{-ictz} \, F(z) dz \\
& =   -\frac{1}{ict}\,e^{-ictz} \,F(z) \Big |_{z=-R}^{R} +  \frac{1}{ict}\, \int_{-R}^{R} e^{-ictz} \partial_z F(z) \, dz.
\end{split}
 \]
 Using \Cref{CV-t^-n-lemma-technic}, we are able to take the limit $R \to \infty$:
  \[
 \lim_{R \to \infty} \int_{-R}^R e^{-ictz} \, \frac{\psi(z)-\psi(0)}{z}\,dz =   \frac{1}{ict}\, \int_{\mathbb{R}_z} e^{-ictz} \partial_z F(z) \, dz.
 \]
 Next, we continue to integrate by parts using \eqref{eq:boundary} (which gets rid of the boundary terms)
 \[
  \lim_{R \to \infty} \int_{-R}^R e^{-ictz} \, \frac{\psi(z)-\psi(0)}{z}\,dz = \left( \frac{1}{ict} \right)^n \, \int_{\mathbb{R}_z} e^{-ictz} \partial_z^n F(z) \, dz,
  \]
thus we have
 $$
  \left|  \lim_{R \to \infty} \int_{-R}^R e^{-ictz} \, \frac{\psi(z)-\psi(0)}{z}\,dz \right| \leq \left( \frac{1}{ct} \right)^n \norm{\partial_z^n F}_{L^1(\mathbb{R}^2)}.
 $$
 We use \eqref{eq:bilinear} to finish the proof of \Cref{thm:1d_cmain}. Now, we need to prove the technical \Cref{CV-t^-n-lemma-technic}.

\begin{proof}[Proof of \Cref{CV-t^-n-lemma-technic}]
  First we can easily show that
  \begin{align}
  \partial_z^{n} F(z) &= \int_{0}^{1} s^n   \partial_z^{n+1} \psi(zs) ds \label{tech-CV-1}\\
   &=  \int_{0}^{z} \frac{s^n}{z^{n+1}}    \partial_z^{n+1} \psi(s) ds \label{tech-CV-2} .
  \end{align}
  
  Then we show that for all $0 \leq m \leq n+1$ and all $p\in [1, \infty ]$
\begin{equation}\label{proof-psi-y}
\norm{ \partial_z^{m} \psi}_{L^p(\R)} \lesssim || C_f ||_{W^{m,p}(\R_z)} \norm{ g_1}_{H^m} \, \norm{ g_2}_{L^2}.
\end{equation}
Indeed, the Leibniz rule yields
 $$
 \partial_z^{m} \psi (z)= \sum_{j=0}^m \begin{pmatrix} m \\ j \end{pmatrix} \partial_z^j C_f(z) \int_{\R_y} \pa_z^{m-j} g_1(y+z) g_2(y) dy.
 $$
By Cauchy-Schwarz inequality,
$$
|\partial_z^{m} \psi(z) | \leq \left( \sum_{j=0}^m \begin{pmatrix} m \\ j \end{pmatrix} |\partial_z^j C_f(z)| \right) \norm{g_1}_{H^m} \norm{g_2}_{L^2}
$$
thus one gets \eqref{proof-psi-y}.   \\

  Then we show that for all $0 \leq m \leq n$,
  $$
\lim_{|z| \to \infty}   \partial_z^{m} F(z)=0.
$$
Indeed, by the Cauchy-Schwarz inequality, \eqref{tech-CV-2} and \eqref{proof-psi-y},
$$
|   \partial_z^{m} F(z) | \leq \norm{  \partial_z^{m+1} \psi}_{L^2(\R_z)} \left( \int_0^{|z|} \frac{s^{2m}}{|z|^{2(m+1)}} ds \right)^{\frac{1}{2}} = \frac{ \norm{  \partial_z^{m+1} \psi}_{L^2(\R_z)} }{ |z|^{\frac{1}{2}}} \underset{|z| \to \infty}{\longrightarrow} 0.
$$
Finally, for all $1 \leq m \leq n$, \eqref{tech-CV-1} implies
$$
\norm{\partial_z^{m} F (z) }_{L^1(\R^2)} \leq \int_{0}^{1} s^m   \norm{ \partial_z^{m+1} \psi (zs)}_{L^1(\R^2)} ds \leq \norm{ \partial_z^{m+1} \psi}_{L^1(\R^2)} \left(  \int_{0}^{1} s^{m-1} ds \right)
$$
and thus we conclude with \eqref{proof-psi-y}.
 \end{proof}  

Natural generalizations of the dynamics in \eqref{eq:1d_cmain_fourier}-\eqref{eq:1d_cmain} to higher dimensions could be obtained in two ways. Firstly, the product $c\,\pa_k \widehat{u}$ in the transport equation \eqref{eq:1d_cmain_fourier} could be generalized to a scalar product of a given unit vector $e\in \R^d$ with the gradient $\nabla_k \widehat{u}(t,k)$.
Unfortunately, this puts too much weight on the constant vector $e$ and results in an obvious statistical anisotropy in physical space. For applications to turbulence, we must require that statistical laws are not only invariant by translation, but also under rotation (i.e. statistical isotropy), as commonly observed in laboratory and numerical experiments, and as expected from a physical point of view. 
Another option would be to replace the multiplication by $ix$ in the physical space formulation \eqref{eq:1d_cmain} by the multiplication by $i|x|$, where $|x|$ is the modulus of $x\in\R^d$. 
Once again, this would introduce anisotropy in the system, and more importantly, it would break statistical homogeneity even in dimension $d=1$. One may check these claims directly using the exact solution \eqref{eq:1d_cmain} to compute the covariance function at a given time $t$ and  any two positions $x,y$ (see \cite{ACM} for details). Indeed, this covariance function eventually depends on the difference $|x|-|y|$, and not on $x-y$ as would be desirable.
Beyond these issues, none of these propositions would ensure that a given real-valued initial condition $u(0,x)\in \R$ gives rise to a real-valued solution $u(t,x)\in\R$ at all future times. In other words, in order to construct a real-valued solution in physical space, one needs to propose a dynamical picture able to preserve the Hermitian symmetry of the Fourier transform $\widehat{u} (t,k)$, as does the dynamics in \eqref{eq:1d_rmain_fourier}.

\section{Higher dimensional real fractional gaussian fields: heuristic}\label{sec-heuristic}

In the previous section we gave a rigorous proof of the construction of a dynamical complex white noise. This proof was carried out in physical space. For the dynamical real white noise of \eqref{eq:1d_rmain_fourier}, on the other hand, it is more convenient to think of its wavenumber formulation.
However, even in the case of  a dynamical complex white noise, the solution $u(t,x)$ is not in $L^p(\mathbb{R}_x)$ for any $1 \leq p < \infty$ (see \eqref{notL^p}), hence its Fourier transform is not defined pointwise.
In conclusion, it is difficult to make sense of equation \eqref{eq:1d_cmain_fourier} in wavenumber space. However, we will not concern ourselves with such difficulties in this section, and we will work as if the solution to \eqref{eq:1d_cmain_fourier} were well-defined pointwise: we refer to Section~\ref{sec-rigorous} for a rigorous analysis. As we will later see, working in wavenumber space is very convenient to formally show that \eqref{eq:1d_rmain_fourier} builds a dynamical real white noise, as well as to extend this construction to $d$-dimensional fractional Gaussian fields.

\subsection{A transport equation in wavenumber space}

We propose the following initial value problem as a generalization of \eqref{eq:1d_rmain_fourier}
\begin{equation}\label{eq:maineq-div}
\begin{cases}
\pa_t \widehat{u}(t,k) +  \vdiv_k \left(  \frac{c k}{| k |} \, \widehat{u}(t,k) \right) + c \displaystyle \frac{H+ \frac{1}{2}}{|k|} \,  \widehat{u}(t,k) = \widehat{f}(t,k) & t>0, k \in \mathbb{R}^d,  |k| >\kappa >0,\\
\widehat{u}(t, k)=0 &  t>0, k \in \mathbb{R}^d, |k| \leq \kappa, \\
\widehat{u}(0,k)=0.
\end{cases}
\end{equation}
where $c,\kappa >0$ are fixed, $H$ is a real constant (which will be eventually connected with the H\"older exponent of the solution),
and $\vdiv_k$ stands for the usual divergence operator in wavenumber space.
As one can easily verify with a few vector calculus identities,
\begin{equation}\label{vector-identities}
 \vdiv_k \left(  \frac{ck}{| k |}\,  \widehat{u}(k) \right) + c \, \displaystyle \frac{H+ \frac{1}{2}}{|k|}\, \widehat{u}(k) = c \, \partial_{|k|} \, \widehat{u}(k) + c \, \frac{H+d-\frac{1}{2}}{|k|}\, \widehat{u}(k)
 \quad \text{with} \quad \partial_{|k|}:= \frac{k}{|k|} \cdot \nabla_k.
\end{equation}

When $d=1$ and $H = -1/2$, one recovers \eqref{eq:1d_rmain_fourier} from the above. The initial value problem \eqref{eq:maineq-div} can be regarded as a conservation law in wavenumber space with a source term $\widehat{f}(t,k)$, a damping term $(H+ 1/2) \,  \widehat{u}(t,k)/|k|$ and Dirichlet boundary conditions at the sphere $|k|= \kappa$. 

Equation \eqref{eq:maineq-div} has been thoroughly studied when $\widehat{f}(t,k)$ is regular enough (see \cite{Bressanone}, \cite{guy}), but in our case $\widehat{f}(t,k)$ is too rough for such classical results to be applicable. In particular, we will assume that the forcing term $ \widehat{f}$ satisfies
\begin{equation}\label{eq:cond_f_FT}
\E [\widehat{f}(t,k_1)\overline{\widehat{f}(s,k_2)}] =  \widehat{C_f}(k_1) \, \delta_{t-s}\, \delta_{k_1-k_2} ,
\end{equation}
where $\widehat{C_f}(k)$ is radial and null in the ball $|k| <\kappa$. Condition \eqref{eq:cond_f_FT} formally follows from considering the Fourier transform of a white noise in time satisfying
\begin{equation}\label{eq:corr_f_cont}
\E [ f(t,x) f(s,y)] = \delta_{t-s}\, C_f (x-y).
\end{equation}

\begin{rk}\label{rk-epsilon=0-1}
As part of equation \eqref{eq:maineq-div} one has the following technical condition
\begin{equation}\label{cond-eps}
\widehat{u}(t, k)=0 \qquad  t>0,\quad k \in \mathbb{R}^d,\quad |k| \leq \kappa,
\end{equation}
for $\kappa >0$. 
Indeed, if $|k|=0$, then $\vdiv_k \left( \frac{k}{| k |} \, \cdot \right)$ and $\frac{H+ \frac{1}{2}}{|k|}$ are not well-defined.
This condition \eqref{cond-eps} implies in particular that $ \mathcal{F}^{-1} \widehat{u}$ has null spatial average, whenever defined. It might be possible to make sense of the problem \eqref{eq:maineq-div} for $\kappa=0$ by adequately changing condition \eqref{cond-eps}, imposing $\widehat{f}(t,k=0)=0$ and an appropriate behaviour near $k=0$, but this is outside the scope of this paper.
\end{rk}

\subsection{Asymptotic behavior: power-law}

In this section, we show that the two-point correlation of the solution to \eqref{eq:maineq-div} displays a power-law behavior. In our first result (\Cref{heuristic-theorem}), we obtain this asymptotic behavior as $t\to\infty$ under fairly mild assumptions on the forcing $\widehat{f}$. Under stronger assumptions on $\widehat{f}$, we derive a second result (\Cref{prop-window-power-law}) showing this power law behavior in finite time.
Among other things, such power laws are important because their exponent determines the H\"older regularity of the solution in physical space (should it be possible to take the inverse Fourier transform). 
The main idea of the heuristic proof of our desired results is to perform a change of variables to rewrite \eqref{eq:maineq-div} as a 1D transport equation with respect to $|k|$ and parametrized by the \enquote{angular variable} $\frac{k}{|k|}$. Such an equation admits an explicit solution that we will exploit.
We will further discuss such consequences in \Cref{sec-rigorous}.

\begin{thm}[Heuristic version]\label{heuristic-theorem}
Let the forcing $\widehat{f}$ satisfy \eqref{eq:cond_f_FT} in such a way that $\widehat{C_f}(k)=\psi(|k|)$ is radial, non-negative, non identically null, with~$s^{2H+d}  \, \psi(s) \in L^1(\mathbb{R}_s^+)$. Furthermore, we assume that~$\psi$  is null when $|k| < \kappa$.
Then, \eqref{eq:maineq-div} admits a solution that satisfies the following asymptotic behavior
$$
\lim_{t \to \infty} \E [\widehat{u}(t,k)\overline{\widehat{u}(t,k')}] =  |k|^{-(2H+d)} \left( C(d,H) - \Psi_{d,H}(|k|)  \right) \delta_{k- k'}
$$
where
\begin{equation}\label{eq:BigPsi}
\Psi_{d,H} (|k|):= \frac{1}{c} \int_{|k|}^{\infty} s^{2H+d}  \, \psi(s) \, ds.
\end{equation}
is a positive non-increasing absolutely continuous function and
\begin{equation}\label{def:C}
C(d,H):=  \Psi_{d,H} (0) >0.
\end{equation}
\end{thm}

As we have already pointed out in \Cref{rk-regular correction}, the function $\Psi_{d,H}$ can be seen as lower-order correction in comparison with the Dirac distribution. Moreover, if we assume that $\psi$ is fast-decaying, then $\Psi_{d,H}$ will also be fast-decaying.
 
\begin{proof}[Heuristic proof]
In order to find a formal solution to \eqref{eq:maineq-div}, we let 
$$
\widehat{v}(t,k):=|k|^{H+d-\frac{1}{2}} \widehat{u}(t,k)
, \quad \widehat{g}(t,k):= |k|^{H+d-\frac{1}{2}}\widehat{f}(t,k).
$$ 
Next we rewrite \eqref{eq:maineq-div} in terms of $\widehat{v}$, namely
\begin{equation}\label{eq:maineq2}
\begin{cases}
\pa_t \widehat{v}(t,k) +c \, \pa_{|k|}  \widehat{v}(t,k)  = \widehat{g}(t,k) & t>0, |k| >\kappa,\\
\widehat{v}(t, k)=0 &  t>0, |k| \leq \kappa, \\
\widehat{v}(0,k)=0,
\end{cases}
\end{equation}
where we use the fact that $\partial_{|k|}= \frac{k}{|k|} \cdot \nabla_k$.

Equation \eqref{eq:maineq2} can be regarded as a 1D transport equation with respect to $|k|$ and parametrized by the \enquote{angular variable} $k/|k|$.  Thus it is easy to give an explicit solution:
\begin{equation}\label{eq:maineq2-sol}
 \widehat{v}(t,k) = \int_{\left(t- \frac{|k| - \kappa}{c}\right)_{+}}^{t}  \widehat{g}\left(s,(|k|-ct+cs)\, \frac{k}{|k|}\right)\, ds.
\end{equation}
Next we compute the correlations of $\widehat{v}(t)$ using those of $\widehat{f}$. We have
\begin{multline*}
\E [\widehat{v}(t,k_1)\overline{\widehat{v}(t,k_2)}] 
= \\
\int_{\left(t- \frac{|k_1| - \kappa}{c}\right)_{+}}^{t} \int_{\left(t- \frac{|k_2| - \kappa}{c}\right)_{+}}^{t}
\delta_{cs_1-cs_2}
\delta_{\frac{|k_1|-ct+cs_1}{|k_1|}\, k_1-\frac{|k_2|-ct+cs_2}{|k_2|}\, k_2}  \, \widehat{C_g}\left(\frac{|k_1|-ct+cs_1}{|k_1|}\, k_1\right) \, ds_1 ds_2.
\end{multline*}
We assume that we can write
\begin{equation*}
\delta_{cs_1-cs_2}
\delta_{\frac{|k_1|-ct+cs_1}{|k_1|}\, k_1-\frac{|k_2|-ct+cs_2}{|k_2|}\, k_2} =
\delta_{cs_1-cs_2}
\delta_{\frac{|k_1|-ct+cs_1}{|k_1|}\, k_1-\frac{|k_2|-ct+cs_1}{|k_2|}\, k_2}
\end{equation*}
even if this not mathematically rigorous. Moreover, by the change of variables from cartesian to polar coordinates
$$
\delta_{\frac{|k_1|-ct+cs}{|k_1|}\, k_1-\frac{|k_2|-ct+cs}{|k_2|}\, k_2} =  \left( \frac{|k_1|-ct+cs}{|k_1|} \right)^{-(d-1)} \delta_{k_1-k_2}.
$$
This is possible since the Jacobian $ \left( \frac{|k_1|-ct+cs}{|k_1|} \right)^{-(d-1)}$ has no singularities in the region of integration thanks to $\kappa >0$. As a result, one obtains
$$
\E [\widehat{v}(t,k)\overline{\widehat{v}(t,k')}] 
= \left(  \int_{\left(t- \frac{|k| - \kappa}{c}\right)_{+}}^{t} \left(\frac{|k|-ct+cs}{|k|}\right)^{-(d-1)} 
\,  \widehat{C_g}\left(\frac{|k|-ct+cs}{|k|}\, k\right) \, ds \right) \delta_{k- k'} .
$$
This immediately implies
$$
\E [\widehat{u}(t,k)\overline{\widehat{u}(t,k')}] 
=  |k|^{-(2H+d)} \left( \int_{\left(t- \frac{|k| - \kappa}{c}\right)_{+}}^{t} \left(|k|-ct+cs\right)^{2H+d} 
\,  \widehat{C_f}\left(\frac{|k|-ct+cs}{|k|}\, k\right) \, ds \right) \delta_{k- k'} .
$$
The change of variables $s \mapsto |k_1|-ct+cs$ yields
\[
\begin{split}
 \E [\widehat{u}(t,k)\overline{\widehat{u}(t,k')}] 
 = & \  |k|^{-(2H+d)}  \left( \chi_{|k|>ct+\kappa} \int_{|k|-ct}^{|k|}s^{2H+d} 
\,  \widehat{C_f}\left(\frac{s}{|k|}\, k\right) \, \frac{ds}{c}  \right) \delta_{k- k'} \\
&  +   |k|^{-(2H+d)}  \left( \chi_{|k| \leq ct+ \kappa} \int_{\kappa}^{|k|} s^{2H+d}
\,  \widehat{C_f}\left(\frac{s}{|k|}\, k\right) \, \frac{ds}{c} \right) \delta_{k- k'}
\end{split}
\]
where $\chi_A$ denotes the characteristic function of the set $A$.
Remember that $\psi (|k|):=\widehat{C_f}(k)$ since~$C_f$ is radial. Using \eqref{eq:BigPsi}, we may rewrite
\begin{equation}\label{cor-u-heuristic-chi}
\begin{split}
\E [\widehat{u}(t,k)\overline{\widehat{u}(t,k')}] =  |k|^{-(2H+d)}  \left( \chi_{|k|>ct+\kappa} \right. &\ \left[  \Psi_{d,H}(|k|-ct)- \Psi_{d,H} (|k|) \right]\\
& + \left. \chi_{|k| \leq ct+\kappa} \left[ \Psi_{d,H}(\kappa) - \Psi_{d,H} (|k|) \right] \right) \delta_{k- k_2}\ .
\end{split}
\end{equation}
We conclude by taking the limit $t \to \infty$ and noting that $\Psi_{d,H}(\kappa)= \Psi_{d,H}(0)$.
\end{proof}


\begin{rk}\label{rk-F}
In the previous proof, we have shown that (see \eqref{cor-u-heuristic-chi}) for any fixed $t >0$
$$
 \E [\widehat{u}(t,k)\overline{\widehat{u}(t,k')}] =   |k|^{-(2H+d)} \, F(t,|k|) \, \delta_{k-k'} 
 $$
 where $F$ is a non-negative absolutely continuous function in $t$ and $|k|$ that admits the following expression
\begin{equation}\label{F-general-psi}
  F(t,|k|) = 
  \begin{cases}
 \Psi_{d,H}(\kappa)-\Psi_{d,H} (|k|) , \, & |k| < ct+\kappa, \\
\Psi_{d,H}(|k|-ct)-\Psi_{d,H} (|k|), \,  & |k| > ct+\kappa.
  \end{cases}
\end{equation}
In particular, for any time $ t_1 \leq t_2$, $F(t_1,|k|)=F(t_2,|k|)$ for all $|k| \leq ct_1+ \kappa$.
As $t$ grows the window of $|k|$ where  $F(t,|k|)$ is stationary (and displays a power-law behavior) grows too (see figure \ref{fig:F(t,|k|)}). 
As a result, it is not necessary to wait until $t =\infty$ in order to observe the power-law. This is helpful when performing numerical simulations.
\end{rk}

\begin{rk}\label{rk-limit-order}
The limits $t \to \infty$ and $ | k |  \to \infty$ in \eqref{F-general-psi} don't commute. Indeed, one gets
$$
\lim_{t \to \infty}  \E [\widehat{u}(t,k)\overline{\widehat{u}(t,k')}] 
\underset{|k | \to \infty}{\sim}  |k|^{-(2H+d)} \, C(d,H) \,  \delta_{k- k'},
$$
whereas for any fixed time $t>0$
$$
\lim_{|k| \to \infty}  \E [\widehat{u}(t,k)\overline{\widehat{u}(t,k')}] =0.
$$
\end{rk}

In simulations, it is usual to look for the power-law for large $|k|$. But, as pointed out in the previous remark, if we don't wait long enough then we will see nothing. The next proposition shows that if the two-point correlation of the source is spectrally located in the ball $|k| \leq r_\ast$, then it is enough to wait until $t\geq r_\ast$ to observe the desired power-law.

\begin{prop}\label{prop-window-power-law}
Let $\psi$ supported in $[\kappa, r_\ast]$ for $r_\ast > \kappa >0$. For any $t > r_\ast$ fixed, the function $F$ given in \Cref{rk-F} satisfies
 $$
 F(t,|k|)=
 \begin{cases}
 0 & | k | \in (0, \kappa) \\
 \text{stationary in $t$ and non-decreasing in $|k|$} &  | k | \in (\kappa, r_\ast), \\
\text{constant}  &   | k | \in ( r_\ast,ct+\kappa), \\
 \text{non-decreasing in $t$ and non-increasing in $|k|$}  &  | k | \in (ct+\kappa, ct+r_\ast ), \\
0 &  | k | \in (ct+r_\ast, \infty).
 \end{cases}
$$
\end{prop}

 \begin{figure}
    \centering
 \includegraphics[width=17cm]{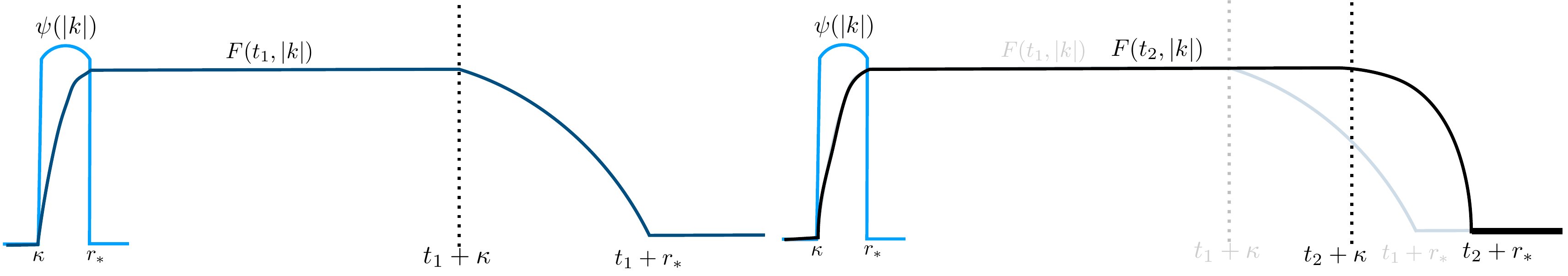}
    \caption{The functions $|k| \mapsto F(t_1,|k|)$ and $|k| \mapsto F(t_2,|k|)$ for $t_1 < t_2$ and $\psi$ supported in $[\kappa, r_\ast].$}\label{fig:F(t,|k|)}
  \end{figure}

\begin{rk}\label{rk-epsilon=0-2}
Since $\psi$ vanishes in the ball $|k|  < \kappa$,
$$
\Psi(0)=\Psi(\kappa) =  \Psi(|k|-ct) \quad \text{for} \quad ct \leq |k| \leq  ct + \kappa,
 $$
 and thus the expression
 $$
   F(t,|k|) = 
  \begin{cases}
 \Psi_{d,H}(0) - \Psi_{d,H} (|k|),  \, & |k| < ct, \\
\Psi_{d,H}(|k|-ct)- \Psi_{d,H} (|k|) , \,  & |k| > ct.
  \end{cases}
  $$
is independent of $\kappa$. On this note, see \Cref{rk-epsilon=0-1}.
\end{rk}


In order that the solution of \eqref{eq:maineq-div-viscous} be closer to the experimental energy-spectrum picture (see beginning of section \ref{Sec:MainResults}), one can add some viscosity in the equation. More precisely, one can consider $f- \nu \Delta u$ instead of a sole forcing term $f$. 
\begin{thm}[Heuristic version]\label{heuristic-theorem_nu}
Let $\nu \geq 0$ and
let the forcing $\widehat{f}$ satisfy \eqref{eq:cond_f_FT} such that $\widehat{C_f}(k)=\psi(|k|)$ is radial, non-negative, non identically null, with~$s^{2H+d}  e^{\frac{8 \pi^2\nu}{3c} s^3} \, \psi(s) \in L^1(\mathbb{R}_s^+)$ and $\psi$   null in the ball $|k| < \kappa$.
Then 
\begin{equation}\label{eq:maineq-div-viscous}
\begin{cases}
\pa_t \widehat{u}_{\nu}(t,k) +  \vdiv_k \left(  \frac{c k}{| k |} \, \widehat{u}_{\nu}(t,k) \right) + \left( c \displaystyle \frac{H+ \frac{1}{2}}{| k|} \, + \nu |2 \pi k|^2 \right) \widehat{u}_{\nu}(t,k)  = \widehat{f}(t,k) 
& t>0, k \in \mathbb{R}^d,  |k| >\kappa >0,\\
\widehat{u}_{\nu}(t, k)=0 &  t>0, k \in \mathbb{R}^d, |k| \leq \kappa, \\
\widehat{u}_{\nu}(0,k)=0 
\end{cases}
\end{equation}
 admits a solution 
that satisfies for any fixed $t >0$
$$
 \E [\widehat{u}_{\nu}(t,k)\overline{\widehat{u}_{\nu}(t,k')}] =   |k|^{-(2H+d)} e^{- \frac{8 \pi^2 \nu}{3c} |k|^3} \, F_{\nu}(t,|k|) \, \delta_{k-k'} 
 $$
 where $F_{\nu}$ is a non-negative absolutely continuous function in $t$ and $|k|$ that admits the following expression
\begin{equation}\label{F-general-psi-viscous}
  F_{\nu}(t,|k|) = 
  \begin{cases}
 \Psi_{d,H,\nu}(\kappa)-\Psi_{d,H,\nu} (|k|) , \, & |k| < ct+\kappa, \\
\Psi_{d,H,\nu}(|k|-ct)-\Psi_{d,H,\nu} (|k|), \,  & |k| > ct+\kappa.
  \end{cases}
\end{equation}
and where
\begin{equation}\label{eq:BigPsi_nu}
\Psi_{d,H,\nu} (|k|):= \frac{1}{c} \int_{|k|}^{\infty} s^{2H+d}  e^{\frac{8 \pi^2 \nu}{3c} s^3} \, \psi(s) \, ds.
\end{equation}
is a positive non-increasing absolutely continuous function.
\end{thm}

\begin{proof}
In order to find a formal solution to \eqref{eq:maineq-div-viscous}, we let 
$$
\widehat{v}(t,k):=|k|^{H+d-\frac{1}{2}} e^{\frac{4 \pi^2 \nu}{3c} |k|^3} \widehat{u}(t,k)
, \quad \widehat{g}(t,k):= |k|^{H+d-\frac{1}{2}} e^{\frac{4 \pi^2 \nu}{3c} |k|^3} \widehat{f}(t,k)
$$ 
and rewriting \eqref{eq:maineq-div-viscous} for $\widehat{v}$, we find that $\widehat{v}$ solves \eqref{eq:maineq2}. Thus the proof is similar to the one of \Cref{heuristic-theorem}.
\end{proof}

Note that in the presence of viscosity, the limits $\nu \to 0$ and $t \to \infty$ commute. 

\subsection{Main issues with the heuristic proof}

The calculations in this section provide invaluable intuition regarding the construction of an equation whose solution displays a power-law behavior as well as the right parameters involved. Nevertheless, there are several steps in our calculations that are difficult to justify from a rigorous mathematical viewpoint (if possible at all). More precisely:
\begin{enumerate}
\item {\it{Fourier-transform.}} As we already mentioned, the Fourier transform of a rough non-decrea\-sing source (and therefore the solution of equation \eqref{eq:maineq-div}) is unlikely to be defined pointwise.
\item{\it{Transport equation in wavenumber space.}} We have never showed that \eqref{eq:maineq-div} is well-posed. Even if it were in some weak sense, it is unclear whether one can define the trace of the solution at the boundary $|k|=\kappa$ or to define the domain of the operator $\vdiv_k \left( \frac{k}{|k|} \cdot \right)$.
\item{\it{Solution of a stochastic transport equation.}} In the simplest case, \eqref{eq:maineq-div} becomes \eqref{eq:maineq2} (for~$H=\frac{1}{2}-d$) whose solution we gave in \eqref{eq:maineq2-sol}, namely
 \begin{equation}\label{eq:maineq3-sol}
 \widehat{u}(t,k) = \int_{(ct- (|k| - \kappa))_{+}}^{t}  \widehat{f}\left(s,(|k|-ct+cs)\, \frac{k}{|k|}\right)\, ds.
\end{equation}
In the case where $f$ is a white noise in time of the form \eqref{eq:rforcing} (in $\mathbb{R}^d$ instead of $\mathbb{R}$), the integral on the right-hand side of \eqref{eq:maineq3-sol} is not a Riemann integral. Worse than that, one would need to define the object $\widehat{f}(t,k)$ and evaluate it along the characteristics of the transport equation. To sum up, one would need to make sense of the ``stochastic integral":
$$
\int_{(t- (|k| - \kappa))_{+}}^{t} \widehat{\varphi}\left((|k|-ct+cs)\, \frac{k}{|k|}\right) \widehat{dW}\left(s,(|k|-ct+cs)\, \frac{k}{|k|}\right).
$$
Moreover, without a rigorous functional setup it is hard to study the important problem of uniqueness of solutions. In fact, it is unclear whether the solution \eqref{eq:maineq3-sol} is unique.
\item {\it{Two-point correlations.}} Even if we can make sense of \eqref{eq:maineq-div} weakly, e.g. as a tempered distribution, the two-point correlation is not necessary well defined. Indeed, it involves a product of distributions whose singular supports might overlap.
\item {\it{Mixing coordinates in Dirac distributions}}. In the proof, we use the following identity
\begin{equation*}
\delta_{s_1-s_2}
\delta_{\tfrac{|k_1|-ct+cs_1}{|k_1|}\, k_1-\tfrac{|k_2|-ct+cs_2}{|k_2|}\, k_2} 
=\delta_{s_1-s_2} \delta_{\tfrac{|k_1|-ct+cs_1}{|k_1|}\, k_1-\tfrac{|k_2|-ct+cs_1}{|k_2|}\, k_2} .
\end{equation*}
In other words, we first apply $\delta_{s_1-s_2}$ to $\delta_{\tfrac{|k_1|-ct+cs_1}{|k_1|}\, k_1-\tfrac{|k_2|-ct+cs_2}{|k_2|}\, k_2}$. Can this be justified?

This is in particular possible if the inner delta were a function, which is of course not the case. Later, in the proof, we perform a change of coordinates to obtain
 $$
\delta_{s_1-s_2}
\delta_{\tfrac{|k_1|-ct+cs_1}{|k_1|}\, k_1-\tfrac{|k_2|-ct+cs_2}{|k_2|}\, k_2} 
= \left( \frac{|k_1|-ct+cs_1}{|k_1|} \right)^{-(d-1)} \delta_{s_1-s_2}  \delta_{k_1-k_2}.
$$
Does the same final result hold if we exchanged the order in which the deltas on the left-hand side are applied?
%
\end{enumerate}

Ultimately, we would also like to write equation \eqref{eq:maineq-div} in physical space and to interpret the power-law in physical space. One first idea is to perform an inverse Fourier transform of \eqref{eq:maineq-div} but the difficulties mentioned above make it complicated.

\section{Higher dimensional real fractional gaussian fields: results}\label{sec-rigorous}

As mentioned in the previous subsection, it is difficult to justify the Fourier transform of a rough source such as the one that interests us. Similarly, the solution to such a problem \eqref{eq:maineq-div} may not admit a Fourier transform either. For that reason, our goal is to try to formulate a weaker notion of \eqref{eq:maineq-div} which is entirely given in physical space, thereby avoiding such problems. 
To do so, we first introduce the Hilbert space
\begin{equation}\label{def:X}
X:= \{ u \in L^2 (\mathbb{R}_x^d) \, | \, \hat{u}(k) = 0 \ \text{for all } |k| <  \kappa \, \} 
\end{equation}
together with the inner product
$$
(u,v):= \int_{\mathbb{R}_k^d \setminus B(0,\kappa)} \widehat{u}(k) \, \overline{\widehat{v}(k)} dk.
$$
Next we introduce the unbounded operator
\begin{equation}\label{def:A}
A : u \mapsto \mathcal{F}^{-1} \left[ c\, \mbox{div}_k \left( \frac{k}{| k |} \, \widehat{u} \right) + c\,  \frac{H+\frac{1}{2}}{|k|} \, \widehat{u} \right]  
\end{equation}
with domain
\begin{equation}\label{def:domain_A}
\begin{split}
D(A) &:=  
\{ u \in X\, | A u  \in X \, \text{and } \, \widehat u |_{|k| = \kappa}=0 \, \}
\\
&=  \{ u \in X\, | \partial_{|k|} \widehat{u} \in L^2 ( \mathbb{R}_k^d \setminus B(0,\kappa)) \, \text{and } \, \widehat u |_{|k| = \kappa}=0 \, \}.
\end{split}
\end{equation}
For any function $u \in X$ such that $\partial_{|k|} \widehat{u} \in L^2 ( \mathbb{R}_k^d \setminus B(0,\kappa))$, the trace in Fourier space $\widehat u |_{|k| = \kappa}$ is well-defined thanks to the Sobolev embedding theorem. This allows us to impose $\widehat u |_{|k| = \kappa}=0$. 
Consequently, the transport equation in wavenumber space \eqref{eq:maineq-div} can be written in physical space after formally applying the inverse Fourier transform:
\begin{equation}\label{eq:maineq-div-phy}
\begin{cases}
\pa_t u + Au = f \qquad \mbox{for}\ t>0,\\
{u}|_{t=0}=0.
\end{cases}
\end{equation}

At this stage, note that for any regular enough forcing $f\in L^1_{\mathrm{loc}}((0,+\infty)_t,X)$, the problem \eqref{eq:maineq-div-phy} is rigorously defined. Its mild solution is given by: 
\begin{equation}\label{duhamel-deterministic}
u(t)= \int_0^t  e^{- (t-s) A} f(s) ds
\end{equation}
where $e^{- tA}: X \to X$ is defined by
\begin{equation}\label{eq:semi-group}
e^{-tA}u_0:=  \mathcal{F}^{-1} \left[ \chi_{|k| > ct + \kappa} \,  \left( \frac{|k|-ct}{|k|} \right)^{H+d-\frac{1}{2}} \widehat{u_0} \left( \frac{|k|-ct}{|k|} k \right) \right].
\end{equation}
Indeed, it is easy to check that \eqref{duhamel-deterministic} solves \eqref{eq:maineq-div-phy}. Moreover, it is also easy to check that if $u_0 \in X \cap \Sc(\mathbb{R}_x^d)$ then $e^{- tA} u_0 \in X \cap \Sc(\mathbb{R}_x^d)$.

Unfortunately, our source $f$ is a white noise in time (and colored in space) and therefore the above considerations do not apply. Rigorously speaking, $f$ is a Gaussian random measure acting on $L^2((0,+\infty)_t)$, i.e.
\begin{equation}\label{eq:cforcing-phy}
\langle f, g \rangle(x) := \int_{(0,+\infty)_t \times \R_y^d} \tau_y \varphi (x)\, g(t) dW(t,y),
\end{equation}
for all test functions $g \in L^2((0,+\infty)_t)$ where $dW$ is a space-time Gaussian real white noise, $\tau_y$ is a translation by $y$ (i.e. $\tau_y g(x)=g(x-y)$) and $\varphi\in  X \cap \Sc (\R_x^d)$ is a real, non-identically null, radial function.

In order to give a meaning to solutions of \eqref{eq:maineq-div-phy} with rough source, we need a weak formulation. To propose one, one first needs to introduce the adjoint of $A$. 

\begin{prop}\label{prop:adjoint}
The adjoint of $A$ in \eqref{def:A} is defined by
\begin{equation}
A^\ast : u \mapsto \mathcal{F}^{-1} \left[ - {\normalfont{c \, \mbox{div}_k}} \left( \frac{k}{| k |}  \, \widehat{u} \right) + c \frac{H+d-\frac{1}{2}}{|k|} \, \widehat{u} \right]  
\end{equation}
and its domain is
$$
D(A^\ast):= \{ u \in X\, | \partial_{|k|} \widehat{u} \in L^2 ( \mathbb{R}_k^d \setminus B(0,\kappa))  \, \} \supsetneq D(A).
$$
\end{prop}

\begin{proof}
Let $u \in D(A)$ and $v \in X$. Using \eqref{vector-identities}, one gets
$$
c^{-1} ( Au,v)_{L^2} = \int_{\mathbb{R}_k^d \setminus B(0,\kappa)} \partial_{|k|} \widehat u(k) \, \overline{\widehat v}(k) \, dk 
+ \int_{\mathbb{R}_k^d \setminus B(0,\kappa)} \frac{H+d-\frac{1}{2}}{|k|}\,\widehat u (k)\, \overline{\widehat v}(k) \, dk.
$$
We perform a change of variables from cartesian to polar coordinates in wavenumber space:
$$
\int_{ \mathbb{R}_k^d} \partial_{|k|} \widehat u(k) \, \overline{\widehat{v}(k)} \, dk = \int_{ \mathbb{R}^+}  \int_{ \mathbb{S}_\theta} \partial_{|k|} \widehat{u}(k) \, \overline{\widehat{v}(k)} \, |k|^{d-1} d \sigma_\theta d |k|.
$$
We now assume that $\partial_{|k|} \widehat{v} \in L^2 ( \mathbb{R}_k^d \setminus B(0,\kappa))$ so that we can integrate by parts in variable $|k|$ and immediately return to cartesian coordinates:
$$
\int_{ \mathbb{R}_k^d} \partial_{|k|} \widehat u \, \overline{\widehat v} \, dk = - \int_{ \mathbb{R}_k^d} \widehat u \, \partial_{|k|} \overline{\widehat v}  \, dk - \int_{ \mathbb{R}_k^d}\frac{(d-1)}{|k|} \widehat u \, \overline{\widehat v}  dk
- \int_{ \mathbb{S}_\theta} \kappa^{d-1}\, \widehat u |_{|k| = \kappa} \, \overline{\widehat v}|_{|k| = \kappa}  d \sigma_\theta.
$$
Since $\widehat u |_{|k| = \kappa} =0$, the previous equality yields
$$
c^{-1} ( Au,v)_{L^2} = - \int_{ \mathbb{R}_k^d} \widehat u \, \partial_{|k|} \overline{\widehat v} + \int_{ \mathbb{R}_k^d}\frac{H+\frac{1}{2}}{|k|} \widehat u \, \overline{\widehat v} \, dk.
$$
Using \eqref{vector-identities} once again, one readily gets
$$
c^{-1} ( Au,v)_{L^2} = \int_{\mathbb{R}_k^d \setminus B(0,\kappa)} \widehat{u}(k) \left[ - \mbox{div}_k \left( \frac{k}{| k |}  \overline{\widehat v}(k) \right) + \frac{H+d-\frac{1}{2}}{|k|} \overline{\widehat v}(k)  \right] \, dk 
$$
for all $u \in D(A)$ and $v \in X$ such that $\partial_{|k|} \widehat{v} \in L^2 ( \mathbb{R}_k^d \setminus B(0,\kappa))$.

Note that it is not necessary to impose $\widehat v |_{|k| = \kappa} =0$ and thus $D(A^{*})$ is strictly larger than~$D(A)$.
\end{proof}

\begin{rk}
In the definition of $D(A)$ the condition $\widehat u |_{|k| = \kappa}=0$ can be interpreted as a boundary condition in wavenumber space. This boundary condition is lost in $D(A^{*})$, which is strictly larger than $D(A)$. It might be possible to take $\kappa=0$ by adequately exchanging $\widehat u |_{|k| = \kappa}=0$ by $\widehat{u}(t,k=0)=0$, and by imposing an appropriate behaviour near $k=0$, but this is outside the scope of this paper. 
\end{rk}

We are ready to give the weak formulation of \eqref{eq:maineq-div-phy}. 

\begin{defn}\label{def:weak_sol}
We say that a stochastic process $u$ is a \emph{weak solution} to \eqref{eq:maineq-div-phy} with $f$ given by~\eqref{eq:cforcing-phy}, if 
$u\in L^1_{\mathrm{loc}} ( [0,\infty)_t \times \R^d)$ almost surely and
\begin{multline}\label{eq:main_weak}
 \int_{\R_x^d} u(t,x)g(t,x) dx + \int_0^t \int_{\R_x^d} u(s,x) (-\pa_s + A^\ast) g(s,x) ds dx \\
 = \int_0^t  \int_{\R_y^d}   \left( \int_{\R_x^d}  \varphi (x-y) g(s,x) dx \right) dW(s,y)
\end{multline}
for all $g \in C^\infty( (0,\infty), D(A^\ast) \cap \Sc (\R_x^d))$, and such that 
\begin{equation}\label{eq:main_weak2}
\int_0^t \int_{\R_x^d} u(s,x)g(s,x) ds dx = 0 
\end{equation}
for all $g \in C^\infty( (0,\infty), \Sc (\R_x^d))$ with $\widehat{g}(t,k)=0$ for all $|k|>\kappa$.
\end{defn}

Note that for a test function $g\in X \cap \Sc (\R_x^d)$,  $g \in D(A^\ast)$ and $A^\ast g \in D(A^\ast) \cap \Sc (\R_x^d)$ and thus the terms on the left-hand side of \eqref{eq:main_weak} are well-defined.

We are ready to give the main result of this paper for positive $H\in (0,1)$.

\begin{thm}\label{thm:main_H_pos}
Let $H\in (0,1)$, and let the forcing $f$ in \eqref{eq:maineq-div-phy} be 
\begin{equation}\label{eq:rforcing-phy}
f(t,x)= \int_{\R_y^d} \tau_y\varphi (x)\, dW(t,y),
\end{equation}
where $dW$ is a space-time Gaussian real white noise, $\tau_y$ is a translation by $y$ (i.e. $\tau_y g(x)=g(x-y)$) and $\varphi\in  \Sc (\R_x^d)$ is a real, non-identically null, radial function such that $\widehat{\varphi}(k)=0$ for all $|k| < \kappa$. Then:
\begin{enumerate}[(i)]
\item Equation \eqref{eq:maineq-div-phy} admits a unique global (in time) weak solution $u(t,x)$ in $C ((0,\infty),(X \cap \Sc (\R_x^d))')$ given by
\begin{equation}\label{eq:mainsol}
u(t,x)=\int_0^t \int_{\R_y^d} e^{- (t-s) A} [\tau_y\varphi](x) \, dW(s,y)
\end{equation}
where $e^{- tA}: X \cap  \Sc (\R_x^d) \to X \cap  \Sc (\R_x^d)$ is defined in \eqref{eq:semi-group}.
\item For each $t>0$ fixed, $u(t,x)$ is a Gaussian process with a.s.\ smooth paths in $x$, and with $\alpha$-H\"older continuous paths in $t$ for any $0<\alpha<1/2$.
\item As $t\rightarrow\infty$, the solution $u(t,x)$ converges in law to a zero-mean Gaussian field $u_{\infty}(x)$ with $\alpha$-H\"older continuous paths in $x$ for any $0<\alpha<H$. 
\item The limiting process $u_{\infty}(x)$ is characterized by the correlations:
\begin{equation}\label{eq:correlations_infty}
\begin{split}
\E [u_{\infty} (x_1) u_{\infty}(x_2)] & =\lim_{t\rightarrow\infty} \E [u(t,x_1) u(t,x_2)]\\
& =  C(d,H) \,\mathcal{K}_H(x_1-x_2) - (\mathcal{K}_H \ast \mathcal{J}_H) (x_1-x_2),
\end{split}
\end{equation}
where 
\[
\mathcal{K}_H := \mathcal{F}^{-1} \left[ \chi_{|k| > \kappa} |k|^{-(2H+d)} \right]
\quad \text{and} \quad
\mathcal{J}_H := \mathcal{F}^{-1} \left[ \chi_{|k| >  \kappa} \,\Psi_{d,H}(|k|) \right]
\]
and where $C(d,H)$ and $\Psi_{d,H}$ are given in Theorem \ref{heuristic-theorem}.
\end{enumerate}
\end{thm}

In the case of non-positive $H$, one needs to be more careful since the solution converges to a very rough object. In this direction, we have the following result:

\begin{thm}\label{thm:main_H_neg}
Let $H\in [-d/2,0]$, and let the forcing $f$ as in \eqref{eq:rforcing-phy}. Then points (i) and (ii) in \Cref{thm:main_H_pos} still hold. Moreover, we have that:
\begin{itemize}
\item [\emph{(iii)}] As $t\rightarrow\infty$, the solution $u(t)$ converges in law 
to a random Gaussian measure $u_{\infty}$ acting on $X \cap \Sc (\R_x^d)$ with zero-mean, i.e. for any $g\in X\cap\Sc (\R^d)$, $\langle u(t),g\rangle$ converges in law to a Gaussian random variable $\langle u_{\infty},g\rangle$ with $\E [ \langle u_{\infty} , g \rangle] = 0$.
\item [\emph{(iv)}]  We have the following asymptotic behavior:
\begin{equation}\label{corr-infty}
\begin{split}
\lim_{t\rightarrow\infty} \E [ \langle u(t), g_1 \rangle {\langle u(t), g_2 \rangle}] =  & \ \E [ \lim_{t\rightarrow\infty} \langle u(t), g_1 \rangle {\langle u(t), g_2 \rangle}] \\
=& \ 
C(d,H)
 \int_{\R_k^d}  \chi_{|k| >  \kappa} \, |k|^{-(2H+d)} \,\widehat{g_1}(k) \,\overline{\widehat{g_2}(k)}  \, dk\\
&  - \int_{\R_k^d}  \chi_{|k| >  \kappa} \, |k|^{-(2H+d)}\, \Psi_{d,H}(|k|)\, \widehat{g_1}(k) \,\overline{\widehat{g_2}(k)} \, dk.
\end{split}
\end{equation}
for any $g_1,g_2\in X\cap\Sc (\R^d)$, where $\Psi_{d,H}$ are given in Theorem \ref{heuristic-theorem}.
\end{itemize}
\end{thm}

Before we prove these theorems, let us make a final comment about the operator $A$. This operator can be regarded as an operator of degree 0 plus a bounded operator of degree $-1$ as  shown in the following proposition.
\begin{prop}
For all $u \in D(A)$ and $v \in D(A^\ast)$, the following equalities hold:
\begin{equation}\label{fibration}
A u (x)  =   \int_{\mathbb{R}_k^d / B_{\kappa} }  a(x,k)  e^{i k \cdot x} \, \widehat{u}(k) \, d k
\quad \text{and} \quad
A^\ast v(x) =  \int_{\mathbb{R}_k^d / B_{\kappa} } a^\ast(x,k) e^{i k \cdot x} \, \widehat{v}(k) d k .
\end{equation}
where
$$
a(x,k) :=-i \frac{k}{|k|} \cdot x + \frac{H+\frac{1}{2} }{|k|} 
\quad \text{and} \quad
 a^\ast(x,k) := i \frac{k}{|k|} \cdot x + \frac{H+d-\frac{1}{2}}{|k|} \cdotp
$$
For any multi-index $(\alpha,\beta)$ with $| \alpha | \geq 1$, one has
$$
|x| |k|^{-|\beta|} \lesssim |  \partial_k^\beta a(x,k)| + | \partial_k^\beta a^\ast(x,k)| \lesssim  |x| |k|^{-|\beta|} + |k|^{-1-|\beta|} 
$$
and
$$
| \partial_x^\alpha \partial_k^\beta a(x,k)| + | \partial_x^\alpha \partial_k^\beta a^\ast(x,k)| \lesssim  |k|^{-|\beta|},
$$
for any $(x,k) \in \mathbb{R}^{2d}$ such that $|k| \geq \kappa$. The implicit constant doesn't depend on $\kappa$, $x$ or $k$.
\end{prop}

\begin{proof}
Let $u \in D(A)$. Using the vector calculus identities \eqref{vector-identities}, one gets
$$
A u = \int_{\mathbb{R}_k^d / B_{\kappa} } \left[ \frac{k}{|k|} \cdot \nabla_k  \widehat{u} + \frac{H+d-\frac{1}{2}}{|k|} \widehat{u} \right]  e^{i k \cdot x} dk.
$$
Integrating by parts
$$
A u = - \int_{\mathbb{R}_k^d / B_{\kappa} } \mbox{div}_k \left( \frac{k}{| k |}  e^{i k \cdot x} \right)  \widehat{u} + \int_{\mathbb{R}_k^d / B_{\kappa} } \frac{H+d-\frac{1}{2}}{|k|} \widehat{u}  e^{i k \cdot x} dk.
$$
As one can easily verify with a few vector calculus identities
\begin{equation}\label{useful-symbol2}
\mbox{div}_k \left( \frac{k}{| k |}  e^{i k \cdot x} \right) =\frac{ i k \cdot x + (d-1)}{|k|} e^{i k \cdot x}
\end{equation}
such that
$$
A u = - \int_{\mathbb{R}_k^d / B_{\kappa} } \frac{ i k \cdot x + (d-1)}{|k|} \widehat{u}  e^{i k \cdot x} dk + \int_{\mathbb{R}_k^d / B_{\kappa} } \frac{H+d-\frac{1}{2}}{|k|} \widehat{u}  e^{i k \cdot x} dk
$$
which gives the expression in \eqref{fibration}. Similar computations hold for $A^\ast$.
\end{proof}

Several recent works \cite{CSR,Carles-quasi-rec,Zworski} show that linear equations involving operators of degree 0 may be useful in showcasing behavior related to turbulence, such as loss of regularity at long time.
In \cite{CSR},  Yves Colin de Verdi{\`e}re and Laure Saint-Raymond in \cite{CSR} consider an equation $\partial_t u + i A u=f$ where~$f$ is an $L^2$ deterministic monochromatic forcing and $A$ is an homogeneous self-adjoint operator of degree 0 that satisfies some technical (but general) assumptions. The authors show that, as~$t \to \infty$, the solution resembles more and more the generalized eigenfunctions of the operator $A$, which do not have $L^2$ finite energy. The fact that $A$ is of order 0 is essential to ensure that its generalized eigenfunctions live in a space akin to $H^{-\frac{1}{2}-}$, which is less regular than $L^2$. R. Carles and C. Cheverry in \cite{Carles-quasi-rec} have also introduced an operator of degree 0 in the context of nuclear magnetic resonance. More precisely, they have shown that for some highly oscillatory source, at long times (i-e diffractive time), the solution produces constructive and destructive interferences which are interpreted as turbulent effects.

It is therefore interesting to compare our model with this existing literature. A first difference is that these other works are developed in deterministic setting. But even if one were to consider our model with a monochromatic deterministic source, it is hard to use similar techniques as those in \cite{CSR}. Firstly, because our operator satisfies $D(A) \subsetneq D(A^\ast)$ and it cannot thus be self-adjoint or skew-adjoint. As a result, one cannot use Mourre's commutator theory as in \cite{CSR} to prove a limiting amplitude principle which allows a nice spectral representation of the solution at infinite time. Secondly, the principal symbol of our operator $A$ (see \eqref{fibration}) grows with respect to $x$. This is not a standard pseudo-differential operator\footnote{A standard operator of degree $m$ is an operator whose symbol $a \in C^\infty(\mathbb{R}^{2d})$ statisfies $$|\partial_x^\alpha \partial_k^\beta a(x,k) | \lesssim (1+ |k|)^{m-|\beta|}$$ for all multi-index $\alpha$ and $\beta$.}  of degree 0. A more profound study of this special operator \eqref{def:A} is therefore interesting, and it is postponed for a future work. 
Nevertheless one can mention some conserved quantities in equation \eqref{eq:maineq-div-phy} when the source is null and initial datum is considered:
\begin{itemize}
\item conservation of volume for all $H$, i-e $$\pa_t \int_{\R_x^d} u(t,x)\, dx=0,$$
\item conservation of $L^2$-norm for $H=-d/2$, i-e $$\pa_t \int_{\R^d_x} |\widehat{u}(t,x)|^2\, dx=0,$$
\item conservation of trace at origin $H=-1/2$, i-e $$\pa_t u |_{x=0}=0.$$
\end{itemize}

\subsection{Well-posedness}

 \subsubsection{Uniqueness}
 
 Suppose that we have two solutions which live a.s. in $L^1_{\mathrm{loc}} ( [0,\infty)_t \times \R^d)$ satisfying \eqref{eq:main_weak}. When tested against test functions whose Fourier transform is supported in $B(0,\kappa)$, both solutions are null and thus they agree. When tested against $g \in C^\infty( (0,\infty), D(A^\ast) \cap \Sc (\R_x^d))$,  their difference $v$ a.s. satisfies
\begin{equation}\label{eq:main_weak_uniq}
 \int_{\R_x^d} v(t,x)g(t,x) dx + \int_0^t \int_{\R_x^d} v(s,x) (-\pa_s + A^\ast) g(s,x) ds dx 
 = 0 \qquad \mbox{for all } t>0.
\end{equation}
If $g_0 \in  D(A^\ast) \cap \Sc (\R_x^d)) $ then 
$$
g_t (s,\cdot) = \mathcal{F}^{-1} \left[ |k|^{-H-\frac{1}{2}} \widehat{g_0} \left( \frac{|k|-c(s-t)}{|k|} k \right) \right]\in C^\infty( (0,\infty), D(A^\ast) \cap \Sc (\R_x^d))
$$
satisfies
$$
(-\pa_s+ A^\ast) g_t = 0 \quad \text{and} \quad g_t |_{s=t}=g_0.
$$
Using this choice of test function, \eqref{eq:main_weak_uniq} yields that a.s. for all $t>0$
$$
\int_{\R_x^d} v(t,x)g_0(x) dx =0.
$$
As a result, almost surely, for all $t$, $v(t,\cdot)=0$ almost everywhere. More precisely, for any two solutions $u_1, u_2$ of our problem, we have:
$$
\mathbb{P} \left[ \forall t>0,\ u_1 (t,x)=u_2(t,x) \ \mbox{almost everywhere in}\ x\right]=1.
$$

\begin{rk} Note that one could obtain a stronger version of uniqueness by slightly changing the definition of weak solution. For instance, one could add the condition that $\E |u(t,x)|^2\in L^1_{\mathrm{loc}}((0,\infty)_t\times\R^d_x)$. Under this new definition one could show that any two solutions $u_1$ and $u_2$ must satisfy
\[
\E |u_1 (t,x)-u_2 (t,x)|=0
\]
$(t,x)$-almost everywhere. This follows directly from the above argument to prove uniqueness together with the fact that the variance is a.e. finite.
\end{rk}

\subsubsection{Existence of a weak solution}

Let us start by giving an expression for $e^{-tA}[\tau_y \varphi](x)$, which will enable us to prove that \eqref{eq:mainsol} does provide a weak solution to \eqref{eq:maineq-div-phy}.

\begin{prop}\label{reg-G}
For all $\varphi\in  X \cap \Sc (\R_x^d)$, the function $G(t,x,y):=  e^{- t A} [\tau_y\varphi](x)$ given by
\begin{equation}\label{eq:green}
\begin{split}
G(t,x,y) & =  \int_{\R^d_k} \chi_{|k| > ct + \kappa} \,  \left( \frac{|k|-ct}{|k|} \right)^{H+d-\frac{1}{2}} \widehat{\varphi} \left( \frac{|k|-ct}{|k|} k \right) e^{- i 2 \pi \left( \frac{|k|-ct}{|k|} \right) k \cdot y +i 2 \pi k\cdot x}\, dk\\
& = \int_{\R^d_k} \left( \frac{|k|}{|k|+ct}\right)^{H+\frac{1}{2}} \widehat{\varphi} (k)\, e^{i 2 \pi \left( \frac{|k|+ct}{|k|} \right) k \cdot x -i 2 \pi k\cdot y}\, dk.
\end{split}
\end{equation}
is a smooth function in all variables. Moreover:
\begin{itemize}
\item For each multi-index $\beta$ and each fixed $y\in\R^d$, $\pa_y^{\beta}G(\cdot,\cdot,y)\in C^1( [0,\infty)_t,  X \cap \Sc( \mathbb{R}_{x}^d))$. Additionally, any Schwartz seminorm of $\pa_y^{\beta}G(\cdot,\cdot,y)$ (with respect to $x$) is locally uniform in $y$.
\item For each $x\in\R^d$, $\pa_x^{\beta} G(\cdot,x,\cdot)\in C^1([0,\infty)_t,\Sc( \mathbb{R}_{y}^d))$. Additionally, any Schwartz seminorm of $\pa_x^{\beta}G(\cdot,x,\cdot)$ (with respect to $y$) is locally uniform in $x$.
\end{itemize}
Finally, its Fourier transform in $y$ is given by
\[
(\mathcal{F}_2 G)(t,x,k)= \left( \frac{|k|}{|k|+ct}\right)^{H+\frac{1}{2}}\, \overline{\widehat{\varphi} (k)}\, e^{-i 2 \pi \left( \frac{|k|+ct}{|k|} \right) k \cdot x}.
\]
\end{prop}
\begin{proof}
The expression \eqref{eq:green} comes from the definition of $e^{-tA}$ in \eqref{eq:semi-group}, and the change of variables~$k\mapsto (|k|+ct)/|k| k$. The second expression directly yields the formula for $(\mathcal{F}_2 G)(t,x,k)$. 

The fact that $G$ is smooth in all three variables follows easily from the integral expression and the fact that we can interchange differentiation and the integral. The latter step will be justified next.
 
Let us show that for any multi-index $\beta$ and any fixed $x$, $\pa_x^{\beta} G(\cdot,x,\cdot)\in C^1([0,\infty)_t,\Sc( \mathbb{R}_{y}^d))$ (the case of fixed $y$ follows a similar proof). Note that any derivative $\pa_y^n \pa_x^{\beta} G(t,x,y)$ results in a factor of $(-i 2 \pi k)^n$, which can be absorbed by $\widehat{\varphi}(k)$, and thus smoothness is no problem. Let us focus on decay. For $|y|\leq 1$, note that
\[
|\pa_x^{\beta} G(t,x,y)| \leq (2 \pi)^{|\beta|} \int_{\R^d_k} \left( \frac{|k|+ct}{|k|} \right)^{|\beta|} \, |\widehat{\varphi} (k)|\, dk.
\]
Recall that the support of $\varphi$ guarantees that there are not problems of integration around $k=0$.

For $|y|\geq 1$ and for any $j\in \{1,\ldots, d\}$,
\[
\begin{split}
\Big | (-iy_j)^n \pa_x^{\beta} G(t,x,y)\Big | & = \Big | \int_{\R^d_k} \left( \frac{|k|}{|k|+ct}\right)^{H+\frac{1}{2}-|\beta|} \widehat{\varphi} (k)\, e^{i 2 \pi \left( \frac{|k|+ct}{|k|} \right) k \cdot x} \pa_{k_j}^n e^{-i 2 \pi k\cdot y}\, dk\Big |\\
& \leq (2 \pi)^{|\beta|} \int_{\R^d_k} \Big | \pa_{k_j}^n \left[\left( \frac{|k|}{|k|+ct}\right)^{H+\frac{1}{2}-|\beta|} \widehat{\varphi} (k)\, e^{i 2 \pi \left( \frac{|k|+ct}{|k|} \right) k \cdot x}\right]\,\Big |\, dk
\end{split}
\]
The integrand of the last expression can be shown to be absolutely integrable on the support of $\widehat{\varphi}$.  Note, however, that one needs to pay some powers of $x$ (which is fixed) to control this term. The implicit constant will therefore depend on such powers, but it is locally uniform in $x$ (i.e. the same constant for all $x$ in a compact set).

This concludes the proof that for each fixed $y\in\R^d$, $G(\cdot,\cdot,y)\in C( [0,\infty)_t,  X \cap \Sc( \mathbb{R}_{x}^d))$.
In order to improve this space to $C^1( [0,\infty)_t,  X \cap \Sc( \mathbb{R}_{x}^d))$, it suffices to differentiate $G$ in time and repeat the above procedure.  

An analogous argument based on integration by parts shows that for each fixed $y\in\R^d$ and each multi-index $\beta$, $\pa_y^{\beta} G(\cdot,\cdot,y)\in C^1( [0,\infty)_t,  X \cap \Sc( \mathbb{R}_{x}^d))$.
\end{proof}

Using these nice properties of $G$, one obtains the following corollary.

\begin{cor} The function $u(t,x)$ defined in \eqref{eq:mainsol} is a well-defined 
Gaussian process. For any $0<\alpha<1/2$, it has a.s.\ $\alpha$-Holder continuous paths with respect to $t$  and a.s.\ smooth paths with respect to $x$. 

In particular, for any multi-index $\beta$, any $t_1,t_2>0$ and any $x_1,x_2\in\R^d$ there exists a (locally uniform) constant $C>0$ such that
\begin{equation}\label{eq:Holder_finite_time}
\E |\pa^{\beta}_x u(t_1,x_1)-\pa^{\beta}_x u(t_2,x_2)|^2 \leq C\left( |t_1-t_2| + |x_1-x_2|^2\right).
\end{equation}
\end{cor}
\begin{proof}
Let $G$ be as in \eqref{eq:green}, then $u$ can be written as
\[
u(t,x)=\int_0^t \int_{\R^d_y} G(t-s,x,y) \, dW(s,y)
\]
 By \Cref{reg-G}, for any fixed $(t,x)$, $G(\cdot,x,\cdot)\in C^1( [0,\infty)_s,  X \cap \Sc( \mathbb{R}_{x}^d))\subset L^2 ([0,t]_s \times \R^d_y)$ and therefore $u$ is a well-defined It\^o process with zero mean and variance 
\[
\E | u(t,x)|^2 = \int_0^t \int_{\R^d_y} |G(s,x,y)|^2 \, dy\, ds.
\]
Next let us study the regularity. Firstly, note that for any multi-index $\beta$
\[
\pa^{\beta}_x u(t,x) = \int_0^t \int_{\R^d_y} \pa^{\beta}_x G(t-s,x,y) \, dW(s,y).
\]
A simple calculation shows that for $t_1<t_2$, 
\begin{equation}\label{eq:change_t}
\begin{split}
\E |\pa^{\beta}_x u(t_1,x) - \pa^{\beta}_x u(t_2,x)|^2 = &\ \int_{0}^{t_1} \int_{\R^d_y} |\pa_x^{\beta} G(t_2-s,x,y)- \pa_x^{\beta} G(t_1-s,x,y)|^2 \, dy\, ds\\
& + \int_{t_1}^{t_2} \int_{\R^d_y} |\pa_x^{\beta} G(t_2-s,x,y)|^2 \, dy\, ds\\
\leq &\ \int_{0}^{t_1} \int_{\R^d_y} |\pa_x^{\beta} G(t_2-t_1+s,x,y)- \pa_x^{\beta} G(s,x,y)|^2 \, dy\, ds\\
& +  |t_1-t_2|\, \norm{\pa_x^{\beta} G(\cdot,x,\cdot)}_{L^{\infty}([0,t_2],L^2(\R^d_y))}^2\\
 \lesssim &\ t_1\, |t_1-t_2|^2\,  \norm{\pa_x^{\beta} \pa_s G(\cdot ,x,\cdot )}_{L^{\infty}([0,t_2],L^2(\R^d_y))}^2 \\
 & +  |t_1-t_2|\,   \norm{\pa_x^{\beta} G(\cdot,x,\cdot)}_{L^{\infty}([0,t_2],L^2(\R^d_y))}^2 
\end{split}
\end{equation}
using the fact that $\pa_x^{\beta} G(\cdot,x,\cdot)\in C^1( [0,\infty)_s,  X \cap \Sc( \mathbb{R}_{x}^d))$.

Similarly, for any $x_1,x_2$, we have that 
\begin{equation}\label{eq:change_x}
\begin{split}
\E |\pa^{\beta}_x u(t,x_1) - \pa^{\beta}_x u(t,x_2)|^2 & = \int_0^t \int_{\R^d_y} |G(s,x_1,y)-G(s,x_2,y)|^2 \, dy\, ds\\
& = \int_0^t \int_{\R^d_y} \Big | \int_0^1 (x_2-x_1)\cdot \nabla_x G(s, (1-\lambda)x_1 + \lambda x_2,y) \, d\lambda \Big |^2 \, dy\, ds\\
& \leq  |x_2-x_1|^2\, \int_0^t \int_{\R^d_y}  \int_0^1 |\nabla_x G(s, (1-\lambda)x_1 + \lambda x_2,y)|^2\, d\lambda \, dy\, ds\\
& \leq  |x_2-x_1|^2\, \int_0^1 \norm{\nabla_x G(s, (1-\lambda)x_1 + \lambda x_2,y)}^2_{L^2 ([0,t]\times \R^d_y)}\, d\lambda 
\end{split}
\end{equation}
The latter factor is finite since $\pa_x^{\beta} G(\cdot,x,\cdot)\in C^1([0,\infty)_t,\Sc( \mathbb{R}_{y}^d))$ locally uniformly in $x$.

One may use \eqref{eq:change_t} and \eqref{eq:change_x} to obtain \eqref{eq:Holder_finite_time}. Given that $\pa_x^{\beta} u(t,x)$ is a Gaussian process, \eqref{eq:Holder_finite_time} immediately implies that for any $m\in\N$,
\[
\E |\pa^{\beta}_x u(t_1,x_1)-\pa^{\beta}_x u(t_2,x_2)|^{2m} \leq C_m\left( |t_1-t_2|^m + |x_1-x_2|^{2m}\right).
\]
The Kolmogorov continuity theorem guarantees that $\pa_x^{\beta} u(t,x)$ is therefore a.s. $\alpha$-H\"older continuous in time for any $0<\alpha<(m-d)/(2m)$.
\end{proof}

We are ready to show that the function $u$ defined in \eqref{eq:mainsol} is indeed a weak solution to \eqref{eq:maineq-div-phy} in the sense given by \Cref{def:weak_sol}.

\begin{proof}[Proof that $u$ is a weak solution]

{\bf Step 1.} We start by checking \eqref{eq:main_weak2}. Choose any $g \in C^\infty( (0,\infty), \Sc (\R_x^d))$ with $\widehat{g}(t,k)=0$ for all $|k|>\kappa$. Then
\begin{multline}\label{eq:Fubini_weak}
\int_0^t \int_{\R_x^d} \left( \int_0^s \int_{\R^d_y} G(s-s',x,y) \, dW(s',y)\right) g(s,x) \, ds dx  =\\
\int_0^t \int_0^s \int_{\R^d_y} \left( \int_{\R_x^d}  G(s-s',x,y) \, g(s,x) \, dx\right) \, dW(s',y)\, ds ,
\end{multline}

In order to justify this, it suffices to show that $G(s-s',x,y) \, g(s,x)$ is absolutely integrable  with respect to the Lebesgue measure in $x\in\R^d$ and $s\in [0,t]$,  and with respect to $d\widetilde{W}(s',y))$ for $s'\in [0,t]$ and $y\in\R^d$.

As proved in \Cref{reg-G}, one may show that $|G(s-s',x,y)|\leq C(s-s',x)\, \langle y\rangle^{-d-1}$ where $C(s-s',x)$ depends on $s-s'$ and $x$ in a polynomial way. To compensate for this growth in $x$, recall that for any $N$, $|g(s,x)|\lesssim_N \langle x\rangle^{-N}$. As a result, $|G(s-s',x,y) \, g(s,x)|\leq C(t,N)\, \langle y\rangle^{-d-1} \langle x\rangle^{-N}$. This guarantees that we may integrate in whichever order we prefer.

In order to check \eqref{eq:main_weak2}, it suffices to show  that for all $s'\in [0,s]$, $s\in [0,t]$ and $y\in\R^d$,
\begin{equation}\label{eq:test_zero}
\int_{\R_x^d} G(s-s',x,y) \, g(s,x) \, dx= 0 .
\end{equation}
This follows from the Plancherel theorem, which allows us to rewrite the integral in \eqref{eq:test_zero} as
\begin{multline*}
\int_{\R_k^d} \widehat{G}(s-s',k,y) \, \widehat{g}(s,k) \, dk=\\
 \int_{\R^d_k} \chi_{|k| > c(s-s') + \kappa} \,  \left( \frac{|k|-c(s-s')}{|k|} \right)^{H+d-\frac{1}{2}} \widehat{\varphi} \left( \frac{|k|-c(s-s')}{|k|} k \right) e^{- i 2 \pi \left( \frac{|k|-c(s-s')}{|k|} \right) k \cdot y} \, \widehat{g}(s,k) \, dk.
\end{multline*}
Note that the support of $\widehat{g}(s,\cdot)$ and that of $\widehat{\varphi}$ do not overlap. Therefore this integral is zero.

{\bf Step 2.} Next we check \eqref{eq:main_weak}. Fix some $g \in C^\infty( (0,\infty), D(A^\ast) \cap \Sc (\R_x^d))$. 
After exchanging the order of integration (which may be justified using \Cref{reg-G} as in Step 1), \eqref{eq:main_weak} is equivalent to:
\begin{multline}\label{eq:test_nonzero}
 \int_{\R_x^d} G(t-s',x,y) \, g(t,x) \, dx + \int_{s'}^t \int_{\R_x^d} G(s-s',x,y) \, (-\pa_s + A^\ast) g(s,x)\, dx\, ds \\
 = \int_{\R_x^d}  \varphi (x-y) g(s',x) dx
\end{multline}
for all $s'\in [0,t]$ and all $y\in\R^d$. Next we may use the pre-dual of $-\pa_s + A^\ast$ to rewrite the second term on the left-hand side. After using the fact that $G(0,x,y)=\varphi(x-y)$, we find that \eqref{eq:test_nonzero} holds if and only if:
\begin{equation}\label{eq:test_nonzero2}
\int_{s'}^t \int_{\R_x^d} (\pa_s + A) G(s-s',x,y) \, g(s,x)\, dx\, ds = 0 \qquad \forall\ s'\in [0,t],\ y\in\R^d.
\end{equation}
Therefore, it suffices to prove that $(\pa_t + A) G(t,x,y)=0$. Taking the Fourier transform in $x$, which is allowed by \Cref{reg-G}, it suffices to check that:
\begin{equation}\label{eq:test_nonzero3}
\pa_t \widehat{G}(t,k,y) +  c\, \mbox{div}_k \left( \frac{k}{| k |}  \widehat{G}(t,k,y) \right) + c \frac{H+d-\frac{1}{2}}{|k|}\,  \widehat{G}(t,k,y)=0
\end{equation}
One may easily check that \eqref{eq:test_nonzero3} does indeed hold by direct calculation using the fact that:
\[
\widehat{G}(t,k,y) =\chi_{|k| > ct + \kappa} \,  \left( \frac{|k|-ct}{|k|} \right)^{H+d-\frac{1}{2}} \widehat{\varphi} \left( \frac{|k|-ct}{|k|} k \right) e^{- i 2 \pi \left( \frac{|k|-ct}{|k|} \right) k \cdot y} .
\]
\end{proof}

\subsection{Correlations for $H\in (0,1)$}

Our first result is an explicit formula for the correlations of our solution at any finite time. Note that this result holds regardless of the value of $H$, because $t>0$ is finite. However, we will only be able to take the limit (as a function) when $H\in (0,1)$.

\begin{prop}\label{thm:correlations}
For any finite time, the correlation function satisfies
$$
\E [ u(t,x_1) u(t,x_2)]= 
 \int_{\R_k^d} |k|^{-(2H+d)}\, F(t,|k|)\, e^{i 2 \pi k(x_1-x_2)} \, dk
$$
where $F$ is the function given in Remark \ref{rk-F}.
\end{prop}
\begin{proof}
For any finite time $t>0$, one has
$$
\E [ u(t,x_1) u(t,x_2)] = \int_0^t \int_{\R^d} e^{- (t-s) A} [\tau_y\varphi](x_1) e^{- (t-s) A} [\tau_y\varphi](x_2) \, dy\, ds.
$$
The change of variables $s \mapsto t-s$ yields
$$
\E [ u(t,x_1) u(t,x_2)]= \int_0^t \int_{\R^d} e^{- s A} [\tau_y\varphi](x_1) e^{- s A} [\tau_y\varphi](x_2) \, dy\, ds.
$$
Using \Cref{reg-G} together with the Plancherel theorem (in the $y$ variable), we obtain
\[
\E [ u(t,x_1) u(t,x_2)]= \int_0^t \int_{\R^d} \left( \frac{|k|}{|k|+cs}\right)^{2H+1}\, |\widehat{\varphi} (k)|^2\, e^{-i 2 \pi \left( \frac{|k|+cs}{|k|} \right) k \cdot (x_1-x_2)} dk\, ds.
\]
Next we perform the change of variables $k\mapsto \frac{|k|}{|k|+cs} \, k$, which yields
\[
\begin{split}
\E [ u(t,x_1) u(t,x_2)] & = \int_0^t \int_{\R^d} \chi_{|k| > cs + \kappa}\, \left( \frac{|k|-cs}{|k|}\right)^{2H+d}\, \Big | \widehat{\varphi} \left(\frac{|k|-cs}{|k|}\, k\right) \Big |^2\, e^{-i 2 \pi k \cdot (x_1-x_2)}\, dk\, ds\\ 
& = \int_{\R^d} e^{-i 2 \pi k \cdot (x_1-x_2)}\, |k|^{-2H-d} \,\left( \int_0^t  \chi_{|k| > cs + \kappa}\, (|k|-cs)^{2H+d}\, \Big | \widehat{\varphi} \left(\frac{|k|-cs}{|k|}\, k\right) \Big |^2\, ds\right) \, dk.
\end{split}
\]
This concludes the proof.
\end{proof}

As a corollary, we take the limit as $t\rightarrow\infty$ and show points (ii) and (iv) of Theorem \ref{thm:main_H_pos}.

\begin{cor}\label{thm:correlations_infty}
Suppose that $H\in (0,1)$. 
Then $u(t,x)$ converges in law to a zero-mean Gaussian field $u_{\infty}(x)$
with the following correlation structure:
\begin{equation*}
\E [u_{\infty} (x_1) u_{\infty}(x_2)]  
 =  \lim_{t\rightarrow\infty} \E [u(t,x_1) u(t,x_2)]  
\end{equation*}
Moreover
$$
\E [u(t,x_1) u(t,x_2)] = C(d,H) \,\mathcal{K}_H(x_1-x_2) - (\mathcal{K}_H \ast \mathcal{J}_H) (x_1-x_2) + O \Big((ct)^{-2H} \Big)
$$
where
\begin{equation}\label{eq:KH_GH}
\mathcal{K}_H := \mathcal{F}^{-1} \left[ \chi_{|k| >  \kappa} \, |k|^{-(2H+d)} \right]
\quad \text{and} \quad
\mathcal{J}_H := \mathcal{F}^{-1} \left[ \chi_{|k| >  \kappa} \,\Psi_{d,H}(|k|) \right]
\end{equation}
and where $C(d,H)$ and $\Psi_{d,H}$ are given in Theorem \ref{heuristic-theorem}. 
\end{cor}
\begin{proof}
Defining 
\[
v(t,x)=\int_0^{t} \int_{\R^d_y} G(s,x,y) \, dW(s,y),
\]
we note that for fixed $t$, $v(t,x)$ and $u(t,x)$ are Gaussian random variables with the same law, since the mean and the correlations coincide.

First we show that
\[
v_{\infty}(x):=\int_0^{\infty} \int_{\R^d_y} G(s,x,y) \, dW(s,y)
\]
is a well-defined Gaussian field. This is due to the fact that for fixed $x\in\R^d$, the function $G(\cdot,x,\cdot)\in L^2 ([0,\infty]_s \times \R^d_y)$. Indeed, following the arguments in \Cref{thm:correlations} one easily finds that
\[
\int_0^{\infty} \int_{\R^d_y} |G(s,x,y)|^2 \, dy\, ds = \int_0^{\infty} \int_{\R^d} \left( \frac{|k|}{|k|+cs}\right)^{2H+1}\, |\widehat{\varphi} (k)|^2\, dk\, ds.
\]

The convergence of $v(t)$ to $v_{\infty}$ is now standard. Note that:
\[
v_{\infty}(x) - v(t,x) = \int_t^{\infty} \int_{\R^d_y} G(s,x,y) \, dW(s,y),
\]
and therefore
\[
\begin{split}
\E |v_{\infty}(x) - v(t,x)|^2 & =  \int_t^{\infty} \int_{\R^d_y} |G(s,x,y)|^2 \, dy\, ds \\
& = \int_t^{\infty} \int_{\R^d} \left( \frac{|k|}{|k|+cs}\right)^{2H+1}\, |\widehat{\varphi} (k)|^2\, dk\, ds\\
& =\frac{1}{2cH}\, \int_{\R^d} (|k|+ct)^{-2H} |k|^{2H+1}\, |\widehat{\varphi} (k)|^2\, dk\\
\end{split}
\]

The same argument allows us to take the limit as $t\rightarrow\infty$ for the two-point correlations of $v(t)$, which coincide with the two-point correlations of $u(t)$ in  \Cref{thm:correlations}, and which yield the correlations of $u_{\infty}$. The fact that we can exchange the limit and the expectation admits the same argument as in Step 2 in the proof of \Cref{thm:1d_cmain}. Indeed, note that the variance $\E |u_{\infty}(x)|^2$ is finite and independent of $x\in\R^d$. 

The rate of convergence of the correlations as $t\rightarrow\infty$ depend directly on the function $F$ in \Cref{thm:correlations}. A careful analysis shows that the rate of convergence is given by a multiple of $(ct)^{-2H}$.
\end{proof}

Now that we know that, whenever $H\in (0,1)$, $u_{\infty}$
is a well-defined Gaussian field indexed by $x\in\R^d$, we may wonder about its continuity and its regularity. In this direction, we have the following result:

\begin{cor}\label{thm:Holder_sol}
Suppose that $H\in (0,1)$. The limit $u_{\infty}=\lim_{t\rightarrow\infty} u(t)$ is a Gaussian field in $x\in\R^d$ which has a modification with a.s. $\alpha$-H\"older continuous paths for any $0<\alpha<H$. More precisely, for all $x, \ell \in\R^d$, the variance of the increment satisfies
\begin{equation}\label{eq:Holder_sol}
\E | \delta_\ell u_{\infty}(x) |^2\leq  \,  \frac{C_{\mathcal{K}}}{H (1-H)} |\ell|^{2H} + C_{\mathcal{J}} |\ell|^2,
\end{equation}
where $C_{\mathcal{K}}$ and $C_{\mathcal{J}}$ only depend of $H$ and $d$ without blowing up when $H$ tends to $0$ or $1$.
\end{cor}
\begin{proof}
By \eqref{eq:correlations_infty}, we have that
\begin{equation}\label{eq:Holder1}
\E | \delta_\ell u_{\infty}(x)|^2=2C(d,H) [\mathcal{K}_H(0) -\mathcal{K}_H(\ell)] -  2 [(\mathcal{K}_H \ast \mathcal{J}_H) (0) - (\mathcal{K}_H \ast \mathcal{J}_H) (\ell)].
\end{equation}

{\bf{Step 1.}} The function $(\mathcal{K}_H \ast \mathcal{J}_H)$ in \eqref{eq:KH_GH} is smooth, since $\Psi_{d,H}(|k|)$ decays as rapidly as desired in $k$. Moreover, we may write:
\begin{equation}\label{eq:Euler}
1- e^{2\pi i k\cdot \ell}=1-\cos\left(2\pi k\cdot \ell \right) - i\, \sin\left(2\pi k\cdot \ell \right).
\end{equation}
Then note that 
\[
\int_{\R^d_k} \sin\left(2\pi k\cdot \ell \right) \chi_{|k| >  \kappa} \, |k|^{2-(2H+d)}\, |\Psi_{d,H}(|k|)|\, dk=0
\]
thanks to the fact that $\chi_{|k| >  \kappa} \, |k|^{-(2H+d)}\, \Psi_{d,H}(|k|)$ is rotationally invariant together with the change of variables $k\mapsto -k$.

As a result,
\[
\begin{split}
|(\mathcal{K}_H \ast \mathcal{J}_H)(0) -(\mathcal{K}_H \ast \mathcal{J}_H) (\ell)| & = \int_{\R^d_k} \left[ 1-\cos\left(2\pi k\cdot \ell \right)\right]\, \chi_{|k| >  \kappa} \, |k|^{-(2H+d)}\, |\Psi_{d,H}(|k|)|\, dk\\
&\leq   \int_{\R^d_k}  \frac{(2\pi |k| \ell)^2}{2}\, \chi_{|k| >  \kappa} \, |k|^{-(2H+d)}\, |\Psi_{d,H}(|k|)|\, dk\\
& \leq 2\pi^2\, |\ell|^2\, \int_{\R^d_k}  |k|^{2-(2H+d)}\, |\Psi_{d,H}(|k|)|\, dk.
\end{split}
\]
The latter integral is finite thanks to the rapid decay of $\Psi_{d,H}(|k|)$. This takes care of the second term on the right-hand side of \eqref{eq:Holder1}.

{\bf Step 2.} The leading term, for the purpose of studying the H\"older regularity, is the first term on the right-hand side of \eqref{eq:Holder1}. We separate the integral into two parts:
\begin{equation}\label{eq:Holder2}
\begin{split}
\mathcal{K}_H(0) -\mathcal{K}_H(\ell) =&\ \int_{\R^d_k}  \chi_{|k| >  \kappa} \, \left( 1- e^{2\pi i k\cdot \ell }\right) \, |k|^{-(2H+d)} \, dk\\
=&\ \int_{\R^d_k}  \chi_{|k|<r} \chi_{|k| >  \kappa} \, \left( 1- e^{2\pi i k\cdot \ell}\right) \, |k|^{-(2H+d)} \, dk\\
& +  \int_{\R^d_k}  \chi_{|k|>r} \chi_{|k| >  \kappa} \, \left( 1- e^{2\pi i k\cdot \ell}\right) \, |k|^{-(2H+d)} \, dk.
\end{split}
\end{equation}
We choose $r=(2\pi |\ell|)^{-1}$; one can carefully check that this choice is optimal. Let us further assume that $|\ell|<\kappa$, the alternative scenario is easier and can be studied separately.

The second term on the right-hand side of \eqref{eq:Holder2} yields:
\[
\Big | \int_{\R^d_k}  \chi_{|k|>r} \chi_{|k| >  \kappa} \, \left( 1- e^{2\pi i k\cdot \ell}\right) \, |k|^{-(2H+d)} \, dk \Big | \leq \Big | \int_{\R^d_k}  \chi_{|k|>r}  \, |k|^{-(2H+d)} \, dk \Big |.
\]
Then, changing to spherical coordinates in the variable $k$, 
\[
\Big | \int_{\R^d_k}  \chi_{|k|>r} \, |k|^{-(2H+d)} \, dk \Big | \lesssim \Big | \int_{|k| > r} |k|^{-(2H+1)} \, d |k| \Big |
\lesssim \frac{1}{H}\,  r^{-2H} \lesssim \frac{1}{H} \, |\ell|^{2H}.
\]
The first term  on the right-hand side of \eqref{eq:Holder2} requires more care. Using \eqref{eq:Euler} as before, note that
\[
\int_{|k|<r}  \sin\left(2\pi k\cdot \ell \right) \, \chi_{|k| >  \kappa} \, |k|^{-(2H+d)} \, dk=0,
\]
as is easily seen from the change of variables $k\mapsto -k$. Therefore it suffices to bound
\[
\begin{split}
\Big | \int_{|k|<r}  \left[ 1-\cos\left(2\pi k\cdot \ell \right) \right] \, \chi_{|k| >  \kappa} \, |k|^{-(2H+d)} \, dk \Big | & \leq\int_{|k|<r}  \frac{(2\pi |k||\ell|)^2}{2} \, \chi_{|k| >  \kappa} \, |k|^{-(2H+d)} \, dk\\
& \lesssim 2\pi^2 |\ell|^2 \int_{\kappa}^{r} |k|^{1-2H}\, d|k|\\
&  \lesssim \frac{1}{2-2H} \, |\ell|^{2H}.
\end{split}
\]
This finishes the proof of \eqref{eq:Holder_sol}. Using this bound on the variance, a standard argument based on the Gaussianity of $u_{\infty}$ and the Kolmogorov continuity theorem show that $u_{\infty}$ has a.s. $\alpha$-H\"older continuous paths for any $0<\alpha<H$.
\end{proof}

\subsection{Correlations for $H\in [-d/2,0]$}
For values of $H$ in this range, we will still show that $u_{\infty}=\lim_{t\rightarrow\infty} u(t)$ exists. However, $u_{\infty}$ cannot be interpreted as a function, in fact its variance is not finite. Instead, $u_{\infty}$ is a random distribution acting on a space of test functions $X\cap S(\R^d)$. Our first result is the analogue of \Cref{thm:correlations} in the context of distributions. 

The rate of convergence depends on some integrability and differentiability properties of the test functions chosen. If one tests against Schwartz functions in $X\cap S(\R^d)$, convergence will happen faster than $t^{-n}$ for any $n\in\N$ (as we will show below). However, one may want to test agains less regular or less decaying functions, in which case this rate of convergence can be quantified. 

Before doing so, one needs to identify a good space of test functions for which ``testing'' is still well-defined. Given that $u(t,x)$ does not generally decay in the $x$-variable, one needs to impose a certain decay on the test functions. For $n\in\N$, let us define the space\footnote{Let us mention that a better space might be possible, in the sense that one might require less decay in physical space. However, this requires a careful study of the right functional space for $u(t,x)$ and $G(s,x,y)$. We leave this study for future research.} $H^{n,n}(\R^d)$, which is the closure of $\Sc (\R^d)$ with respect to the norm
\[
\norm{g}_{H^{n,n}(\R^d)} = \sum_{j=0}^n \norm{ \langle x\rangle^j g}_{H^{n-j} (\R^d_x)}.
\]
It is not hard to show that $H^{n,n}(\R^d)$ is a Banach space algebra living inside $L^2(\R^d_x)$, and thus the Fourier transform is well-defined for such functions. Moreover, $g\in H^{n,n}(\R^d_x)$ if and only if $\widehat{g}\in H^{n,n}(\R^d_k)$ (and both norms are comparable).
%

\begin{prop}\label{thm:correlations_tested}
Let $H\in [-d/2,0]$. For any $t>0$ and any test functions $g_1,g_2 \in H^{n,n} (\R^d)$ with $3(d+1)/2<n$,
the correlations satisfy
\begin{equation}\label{corr-dulalisation}
\begin{split}
\E [ \langle u(t), g_1 \rangle {\langle u(t), g_2 \rangle}] =  & \ 
C(d,H)
 \int_{\R_k^d}  \chi_{|k| >  \kappa} \, |k|^{-(2H+d)} \,\widehat{g_1}(k) \,\overline{\widehat{g_2}(k)}  \, dk\\
&  - \int_{\R_k^d}  \chi_{|k| >  \kappa} \, |k|^{-(2H+d)}\, \Psi_{d,H}(|k|)\, \widehat{g_1}(k) \,\overline{\widehat{g_2}(k)} \, dk
\\
 & +\int_{\R_k^d} \chi_{|k| > ct + \kappa}\, |k|^{-(2H+d)} \, [\Psi_{d,H}(|k|-ct) - \Psi_{d,H}(0) ]\, \widehat{g_1}(k) \, \overline{\widehat{g_2}(k)} \,dk\\
 =  & \ 
C(d,H)
 \int_{\R_k^d}  \chi_{|k| >  \kappa} \, |k|^{-(2H+d)} \,\widehat{g_1}(k) \,\overline{\widehat{g_2}(k)}  \, dk\\
&  - \int_{\R_k^d}  \chi_{|k| >  \kappa} \, |k|^{-(2H+d)}\, \Psi_{d,H}(|k|)\, \widehat{g_1}(k) \,\overline{\widehat{g_2}(k)} \, dk \\
& + C(d,H)\, \left(\frac{1}{ct}\right)^{-(2H+d)-2n} \, \norm{g_1}_{H^{n,n}} \, \norm{g_2}_{H^{n,n}} 
\end{split}
\end{equation}
where $C(d,H)$ and $ \Psi_{d,H}$ are given in Theorem \ref{heuristic-theorem}, and $F$ is given by \eqref{rk-F}. 
\end{prop}
\begin{rk}
If one chooses test functions $g_1,g_2\in X\cap \Sc (\R^d)$, then the above result yields a rate of convergence proportional to $(ct)^{-(2H+d)-2n}$ for $n$ as large as desired.
\end{rk}
\begin{proof}
{\bf Step 1.} Firstly, note that $\langle u(t),g_1\rangle$ is well-defined for each $t>0$. This follows after a similar argument to the one that justifies \eqref{eq:Fubini_weak}, i.e. the fact that $G(s,x,y) g_1 (x)$ is absolutely integrable with respect to the Lebesgue measure in $x$ and with respect $dW(s,y)$. Indeed, we showed that $|G(s,x,y)|\lesssim_t \langle x\rangle^{d+1} \langle y\rangle^{-d-1}$, so one needs only show that $\langle x\rangle^{d+1} |g_1 (x)|\in L^1 (\R^d_x)$. But this follows from the choice of $n$ together with the Cauchy-Schwarz inequality:
\[
\int_{\R^d} \langle x\rangle^{d+1} |g_1 (x)| \, dx \lesssim \norm{g_1}_{H^{n,n}}\, \left(\int_{\R^d} \langle x\rangle^{-d-1} \, dx\right)^{1/2}.
\]

By \Cref{thm:correlations}, we have:
\[
\begin{split}
\E [ \langle u(t),g_1 \rangle  {\langle u(t),g_2 \rangle } ] & =
 \int_{\mathbb{R}_{x_1}^d \times \mathbb{R}_{x_2}^d}  \left( \int_{\R_k^d} |k|^{-2H-d} F(t,|k|) e^{i 2 \pi k(x_1-x_2)} \, dk \right) g_1(x_1) g_2(x_2) \, dx_1 dx_2 \\
& =  \int_{\R_k^d} |k|^{-2H-d} F(t,|k|) \,\widehat{g_1}(k) \, \overline{\widehat{g_2}(k)} dk 
\end{split}
\]
Finally, we use the expression for $F(t,|k|)$ obtained in \Cref{rk-F}. This yields \eqref{corr-dulalisation}.

{\bf Step 2.} In order to obtain an asymptotic expansion in terms of $t$, note that it suffices to study the last term on the right-hand side of \eqref{corr-dulalisation} which, using the definition of $\Psi_{d,H}$, can be rewritten as follows:
\[
-\int_{\R_k^d} \chi_{|k| > ct + \kappa}\, |k|^{-(2H+d)} \, \left( \int_{\kappa}^{|k|-ct} s^{2H+d} \psi (s)\, ds \right) \, \widehat{g_1}(k) \, \overline{\widehat{g_2}(k)} \,dk.
\]
Note that generally, we cannot expect a better bound than 
\[
\Big |\chi_{|k| > ct + \kappa}\, \left( \int_{\kappa}^{|k|-ct} s^{2H+d} \psi (s)\, ds \right) \Big | \leq C(d,H) \lesssim 1.
\]
Certainly, better bounds could be obtained should $\psi$ vanish on a large ball around zero, but this is not the physical setting we are interested in. 


Therefore, by the Cauchy-Schwarz inequality,
\begin{multline*}
\Big |\int_{\R_k^d} \chi_{|k| > ct + \kappa}\, |k|^{-(2H+d)} \, \left( \int_{\kappa}^{|k|-ct} s^{2H+d} \psi (s)\, ds \right) \, \widehat{g_1}(k) \, \overline{\widehat{g_2}(k)} \,dk \Big | \\
\leq C(d,H)\,\int_{\R_k^d} \chi_{|k| > ct + \kappa}\, |k|^{-(2H+d)} \, |\widehat{g_1}(k)| \, |\widehat{g_2}(k)| \,dk\\
\lesssim C(d,H)\, (ct)^{-(2H+d)-2n} \, \norm{g_1}_{H^{n,n}} \, \norm{g_2}_{H^{n,n}} .
\end{multline*}
\end{proof}

Thanks to the formula for the correlations obtained in \Cref{thm:correlations_tested}, we may now study the asymptotic behavior of $u(t)$ as $t\rightarrow\infty$. That is the content of the next result.

\begin{cor}\label{thm:correlations_tested_infty}
Let $H\in [-d/2,0]$ and $n>3(d+1)/2$. The solution $u(t)$ in \eqref{eq:mainsol} converges in law, as $t\rightarrow\infty$, to a random Gaussian measure acting on $X \cap H^{n,n} (\R^d)$. Moreover, the correlation structure of the limiting Gaussian measure is given by:
\begin{equation}\label{corr-infty-2}
\begin{split}
\E [  \langle u_{\infty}, g_1 \rangle {\langle u_{\infty}, g_2 \rangle}]
=  & \  \lim_{t\rightarrow\infty} \E [ \langle u(t), g_1 \rangle {\langle u(t), g_2 \rangle}] \\
=& \ 
C(d,H)
 \int_{\R_k^d}  \chi_{|k| >  \kappa} \, |k|^{-(2H+d)} \,\widehat{g_1}(k) \,\overline{\widehat{g_2}(k)}  \, dk\\
&  - \int_{\R_k^d}  \chi_{|k| >  \kappa} \, |k|^{-(2H+d)}\, \Psi_{d,H}(|k|)\, \widehat{g_1}(k) \,\overline{\widehat{g_2}(k)} \, dk.
\end{split}
\end{equation}
for any $g_1,g_2\in X\cap H^{n,n} (\R^d)$.
\end{cor}

\begin{proof}
The convergence of $\langle u(t),g\rangle$ as $t\rightarrow\infty$ follows from 
a similar argument to the one in the proof of \Cref{thm:correlations_infty}, so let us focus on proving \eqref{corr-infty-2}. The second  equality in \eqref{corr-infty-2} follows by taking $t\rightarrow\infty$ in \eqref{corr-dulalisation}. Finally, the first equality admits the same argument as Step 2 in the proof of \Cref{thm:1d_cmain}, so we omit it.
\end{proof}

\subsection{The case $H=-d/2$}\label{sec-whitenoise}

In this case, note that the limiting Gaussian field $u_{\infty}$ obtained in \Cref{thm:correlations_tested_infty} seems to  differ from the usual white noise measure in two ways:
\begin{enumerate}[(i)]
\item the space of test functions $X\cap \Sc (\R^d)$, which only admits functions whose Fourier transform vanishes in $|k|<\kappa$; and
\item the top order term in \eqref{corr-infty-2}, which involves a cut-off to wavenumbers $|k|>\kappa$.
\end{enumerate}

Point (ii) can be fixed by defining a continuous zero-mean Gaussian field $u_{\mathcal{I}}$ with correlations
\begin{equation}\label{eq:def_ureg}
\begin{split}
\E [u_{\mathcal{I}} (x_1) u_{\mathcal{I}}(x_2)] & =  \mathcal{J}_{-d/2}(x_1-x_2) + \mathcal{F}^{-1}\left[  \chi_{|k| \leq \kappa} \right](x_1-x_2)\\
& =  \int_{\R^d_k} \chi_{|k| >  \kappa} \,\Psi_{d,H}(|k|)\, e^{2\pi i k\cdot(x_1-x_2)} \, dk  + \mathcal{F}^{-1}\left[  \chi_{|k| \leq \kappa} \right](x_1-x_2)\\
& =: \mathcal{I}(x_1-x_2),
\end{split}
\end{equation}
where $\mathcal{J}$ was defined in \eqref{eq:KH_GH}.
Note that $\mathcal{I}(x)$ is a continuous, radial, square-integrable function, and therefore they are much more regular than a delta function. For instance, in the one-dimensional case $d=1$, the second term on the right-hand side of \eqref{eq:def_ureg} corresponds to the sinc function.

Using \eqref{eq:def_ureg}, we may rewrite the correlations of our solution $u_{\infty}$ in \eqref{corr-infty-2} as follows: 
\begin{equation}\label{eq:true_whitenoise}
\begin{split}
 \E [ \langle u_{\infty}, g_1 \rangle {\langle u_{\infty}, g_2 \rangle}]
 & =  \E [ \langle  dW, g_1 \rangle {\langle  dW, g_2 \rangle}]
- \int_{\mathbb{R}^{2d}}  \mathcal{I}(x_1-x_2)\, g_1(x_1) \, g_2(x_2) \, dx_1 dx_2\\
& = \int_{\R^d_x} g_1(x) g_2(x)\, dx- \int_{\mathbb{R}^{2d}} \mathcal{I}(x_1-x_2)\, g_1(x_1) \, g_2(x_2) \, dx_1 dx_2
\end{split}
\end{equation}
for any $g_1,g_2 \in  X\cap \mathcal{S} (\R_x^d)$. This solves point (ii). 

Let us next focus on point (i), regarding the space of test functions. There is a natural way to regard $u(t)$ as a distribution in $\Sc(\R^d)'$ (as opposed to $(X\cap \Sc(\R^d))'$). Fix a test function $h\in \Sc (\R^d)$ such that 
\begin{equation}\label{eq:cond_h}
\begin{cases}
\widehat{h}(k) = 0 \qquad \mbox{if}\ |k|\geq \kappa,\\
\widehat{h}(k) \  \mbox{is radial and takes values in}\  [0,1]\ \mbox{for}\ |k|<\kappa.
\end{cases}
\end{equation}
For any test function $g\in \Sc (\R^d)$, one may decompose
\[
g = \mathcal{L}_1 g+\mathcal{L}_2 g := \mathcal{F}^{-1} [ \,\widehat{g} \, \widehat{h}\, ] + \mathcal{F}^{-1} [ \,\widehat{g} \, (1-\widehat{h}) \,] .
\]
Then one may make a small modification to the notion of solution in \Cref{def:weak_sol}. Indeed, one could define a weak solution to to \eqref{eq:maineq-div-phy} to be a stochastic process $u$ such that $u\in L^1_{\mathrm{loc}} ( [0,\infty)_t \times \R^d)$ almost surely and such that for all $g \in C^\infty( (0,\infty), \Sc (\R_x^d))$
\begin{multline*}
 \int_{\R_x^d} u(t,x) \mathcal{L}_2g(t,x) dx + \int_0^t \int_{\R_x^d} u(s,x) (-\pa_s + A^\ast) \mathcal{L}_2g(s,x) ds dx \\
 = \int_0^t  \int_{\R_y^d}   \left( \int_{\R_x^d}  \varphi (x-y)\, \mathcal{L}_2 g(s,x) dx \right) dW(s,y)
\end{multline*}
and
\[
\int_0^t \int_{\R_x^d} u(s,x)\, \mathcal{L}_1 g(s,x) ds dx = 0 .
\]

Note that the solution found in \eqref{eq:mainsol} is still a solution with this new definition. Moreover, the solution does not depend on the choice of function $h$ as long as it satisfies \eqref{eq:cond_h}. 

This allows us to view $u(t)$ and $u_{\infty}$ as distributions in $\Sc(\R^d)'$, and thus to extend \eqref{eq:true_whitenoise} to test functions in $\Sc(\R^d)$. Then the top order of \eqref{eq:true_whitenoise} does indeed correspond to a white noise in the classical sense.

\section{Numerical Simulations}\label{sec-simulations}

Our goal in this section is to perform numerical simulations to illustrate the theory presented in the previous sections. We recall the continuous problem \eqref{eq:maineq-div-viscous}:
\begin{equation}\label{eq:maineq-div-viscous_simul}
\begin{cases}
\pa_t \widehat{u}(t,k) +  \mathfrak{L}\big( \widehat{u} \big) (t,k) + \mathfrak{D}\big( \widehat{u} \big) (t,k) = \widehat{f}(t,k) 
& \text{ for }  t>0, k \in \mathbb{R}^d,  |k| >\kappa \\
\widehat{u}(t, k)=0 &\text{ for }  t>0, k \in \mathbb{R}^d, |k| \leq \kappa, \\
\widehat{u}(0,k)=0 &\text{ for }  k \in \mathbb{R}^d,\\
\end{cases}
\end{equation}
where
\begin{equation}\label{flux-numeric}
\mathfrak{L}\big( \widehat{u} \big) (t,k):= c\vdiv_k \left(  \frac{ k}{| k |} \, \widehat{u}(t,k)\right)
\quad \text{and} \quad
\mathfrak{D}\big( \widehat{u} \big) (t,k):= \left(c\, \displaystyle \frac{H+ \frac{1}{2}}{|k|} \,  + (2\pi)^2\nu |k|^2\right)\widehat{u}(t,k)
\end{equation}
Here, $H$ is eventually the Hurst exponent of the solution (at infinite time), and $\nu>0$ denotes viscosity. 
The introduction of viscosity is necessary in order to reach a statistically stationary state at a finite time, state in which a statistical analysis is possible. In the inviscid problem proposed in \eqref{eq:maineq-div}, on the other hand, scales as small as $1/ct$ are populated at time $t$, leading to numerical instabilities and a blow up when simulated in a finite periodic box with a finite resolution.

Our numerical method combines ideas from time predictor-corrector schemes \cite{kloeden2011numerical} and pseudo-spectral methods \cite{pope2000turbulent}. In particular, this implies using the Discrete Fourier Transform (DFT), which forces certain choices regarding the discretization. 

We will denote by $\widehat{u}[t,k]$ the (discrete) vector whose continuous counterpart is $\widehat{u}(t,k)$. We denote by $\Delta t >0$ the time stepping and by $\Delta x>0$ the mesh size. The mesh size is also the smallest accessible length scale.

\subsection{Numerical method}

\subsubsection{Discretization}

We discretize our physical space by considering a discrete periodic box $x\in(\Z/N\Z)^d$ of unit length $L_{tot}=1$, using $N=2^n$ collocation points in each direction, with $n\in\N$ where we adopt the convention that $0$ is not a natural number. Therefore, the mesh size is $$\Delta x=L_{tot}/N.$$ The wave vector $k=(k_i)_{1\le i\le d}$ is discretized as 
\[
k_i=[0,1,...,N/2,-N/2+1,-N/2+2,...,-1]\Delta k,
 \]
where the spectral resolution is given by $$\Delta k=1/L_{tot}.$$ This discretization is standard when using the Discrete Fourier Transform (DFT). The choice of starting from $k_i=0$ as the first element of the array is dictated by the convention that is used to define the DFT and its inverse.

In order to discretize the divergence term $\mathfrak{L}\big( \widehat{u} \big)$ in \eqref{flux-numeric}, we introduce the following discretization of derivatives in each component $k_j$ of the wavenumber:
\begin{equation}\label{eq:MeaningDifferentK}
 \partial_{k_j}\widehat{g}[t,k] = \text{DFT}\left[ -2\pi i \tilde{x}_j \text{DFT}^{-1}\left[\widehat{g}[t,k]\right] \right],
 \end{equation}
where the component $\tilde{x}_j =[0,1,...,N/2-1,0,-N/2+1,-N/2+2,...,-1]\Delta x$, considered as a vector and  $\text{DFT}^{-1}\left[\widehat{g}[k]\right]$ as a scalar, corresponds to the component $j$ of the position $x_j$ where the ``Nyquist'' mode has been set to 0 in order to respect the parity of the differentiation. 

\subsubsection{Discretization of the forcing}

In order to produce the forcing term $\widehat{f}[t,k]$, we generate an instance of $N^d$ independent, 
zero-average and unit-variance Gaussian random variables at each time step, which we store in a vector $g[t,x]$. 
Then we weigh them by the appropriate factor $(\Delta x)^{d/2}$, we take the DFT and multiply by the indicator function
 $\chi_{3\le |k|L_{tot} \le 5}$, which ensures in particular that no energy is injected at the mode $k=0$, and mostly at large scales.
Mathematically, one gets
\begin{equation}\label{eq:num_forcing}
\widehat{f}[t,k] = \chi_{3\le |k|L_{tot} \le 5}\, \text{DFT}\left[ (\Delta x)^{d/2}\, g[t,x]\right].
\end{equation}
Finally, we choose $\kappa = 1/L_{tot}$ which is smaller than the smallest non vanishing wavenumber of the source $\widehat{f}[t,k]$ (see Remark \ref{rk-kappa-source}).


\subsubsection{Algorithm}\label{Sec:Algo}

\

\noindent $\bullet$ We pre-compute the forcing $\widehat{f}[t,k]$ as in \eqref{eq:num_forcing}.

\noindent $\bullet$ Initialization step: $\widehat{u}[0,k]=0$.

\noindent $\bullet$ Induction step: given $\widehat{u}[t,k]$, we compute $\widehat{u}[t+\Delta t,k]$ via the following procedure:
\begin{enumerate}
\item {\it{(Prediction step - spatial part)}} Compute $ \mathfrak{L}\big( \widehat{u} \big)[t,k]$ according to \eqref{eq:MeaningDifferentK}, and compute the numerical damping $\mathfrak{D}\big( \widehat{u} \big)[t,k]$.
\item {\it{(Prediction step - temporal part)}} For each $k$ such that $|k| \geq \kappa$, compute the predictor $\widehat{u}^*[t,k]$ according to
\[
\frac{\widehat{u}^*[t,k] - \widehat{u}[t,k]}{\Delta t} +  \mathfrak{L}\big( \widehat{u} \big)[t,k] + \mathfrak{D}\big( \widehat{u} \big)[t,k] = \widehat{f}[t,k] (\Delta t)^{-\frac{1}{2}}\ .
\]
\item {\it{(Correction step - spatial part)}} Compute $ \mathfrak{L}\big( \widehat{u}^* \big)[t,k]$ with \eqref{eq:MeaningDifferentK}, and the numerical damping $\mathfrak{D}\big( \widehat{u}^* \big)[t,k]$.
\item {\it{(Correction step - temporal part)}}  For each $k$ such that $|k| \geq \kappa$, compute the corrector $\widehat{u}[t+\Delta t,k]$ according to
\[
\frac{\widehat{u}[t+\Delta t,k]- \widehat{u}[t,k]}{\Delta t} + \frac{\mathfrak{L}\big( \widehat{u}^* \big)[t,k] + \mathfrak{D}\big( \widehat{u}^* \big)[t,k] + \mathfrak{L}\big( \widehat{u} \big)[t,k] + \mathfrak{D}\big( \widehat{u} \big)[t,k]}{2} =  \widehat{f}[t,k]  (\Delta t)^{-\frac{1}{2}}\ .
\]
\end{enumerate}

\subsection{Discussion of the numerical method} 

\subsubsection{Discretization of the operators $\mathfrak{L}$}

For any choice made for the estimation of the derivatives entering in $\mathfrak{L}\big( \widehat{u} \big)$ (see \eqref{flux-numeric}), the way the wavevector is discretized in a Cartesian fashion is not well adapted to this numerical problem which has natural spherical symmetry. Besides the forcing $ \widehat{f}$, which is only statistically isotropic, the deterministic part of the evolution is spherically symmetric\footnote{
Indeed, all quantities involve the vector amplitude $|k|$. Moreover, $\mathfrak{L}\big( \widehat{u} \big)$ can be written as a radial derivative with respect to $|k|$ (see \eqref{vector-identities}).}.
Nonetheless, the way it is presently discretized in a Cartesian form allows an easy and standard implementation of the DFT (using the Fast Fourier Transform algorithm) such that the solution $u[t,x]$ in physical space, over the discrete set of positions $x=(x_i)_{1\le i\le d}$ with $x_i=[0,1,...,N/2,-N/2+1,-N/2+2,...,-1]\Delta x$, is obtained with an inverse DFT of $\widehat{u}[t,k]$.

As we have already mentioned, we cannot make sense of $\widehat{u}(t,k)$ as a pointwise function of $k$. To this regard, the formulation of the problem in Fourier space needs to be justified and a proper meaning has to be given to the divergence of a rough field, see \Cref{sec-heuristic,sec-rigorous} for full details. For the purpose of numerical simulations we have thus proposed \eqref{eq:MeaningDifferentK} as a numerical realization of the derivatives, with the use of back and forth DFTs and an appropriate multiplication by a space-dependent factor.

\subsubsection{Discretization of the operators $\partial_t$}

Concerning time integration for a given time stepping $\Delta t$, we use an explicit predictor-corrector method with a single independent realization of the forcing, consisting in predicting the solution $\widehat{u}[t+\Delta t,k]$ using an explicit Euler discretization scheme taking previous time step $\widehat{u}[t,k]$ as an initial condition, and correcting by another Euler scheme which weighs the initial condition and the prediction equally.

\subsubsection{Choice of discretization parameters}\label{RK:Parameters}

Let us give the last relevant parameters that we use for our simulations, keeping in mind that the spatial resolution $\Delta x=L_{tot}/N$ is determined by the physical length of the side of the periodic box (henceforth, we choose $L_{tot}=1$ without loss of generality), and the number of collocation points $N$ in each direction.  As we will see, the viscosity $\nu$ is chosen accordingly. 

Firstly, the time step $\Delta t$ has to be chosen. At this stage, a CFL criterion for the dynamics proposed in our algorithm is not clear. Thus for the time being, we assume that a reasonable stability criterion may come from the time derivative and the viscous term. In that spirit, it seems reasonable to consider the numerical stability of the heat equation, which is imposed by the viscous term  and which requires\footnote{A more sophisticated scheme could be written including explicitly the exact solution of the underlying heat equation in the time stepping; this is known as an exponential scheme and would allow for $\nu\Delta t$ of order $\Delta x$.} that $\nu\Delta t<(\Delta x)^2/2$ (see for instance Ref. \cite{press2007numerical}).

Forthcoming simulations will be done for very high values of the number of collocation points $N$ in order to consider small values for the viscosity $\nu$. Even in the most comfortable situation where $\nu\Delta t$ is of order $\Delta x$, the numerical cost will eventually become prohibitive. For this reason, we choose a time step independent of resolution and dimension $d$, and given by $\Delta t=5 \times 10^{-3}$. Additional simulations (data not shown) for $d=1$ and $d=2$ using $\Delta t=(\Delta x)^2/2$ and for moderate numbers of collocation points $N$ have given similar numerical results as the ones obtained with $\Delta t=5 \times10^{-3}$. It is unclear at this stage why this chosen value of the time step does not lead to numerical instabilities, although we could invoke the fact that viscosities will be chosen very small.

\subsection{Simulations}
 
 \subsubsection{Determination of viscosity and averaging procedure} \label{Sec:DeterViscAver}
 
 All predictions that have been made in the theoretical sections concern the statistical behavior of the solution $u(t,x)$ of the continuous problem recalled in \eqref{eq:maineq-div-viscous}, with or without viscosity. From a numerical point of view, we need to define a time $T_\ast$ at which the system has reached a statistically stationary regime, in which mathematical expectations will be approximated by an empirical average in time. Because of the cascading process of energy will eventually populate modes at higher and higher wavenumber $k$ as time goes on, we introduce viscosity in order to damp all the energy once it has reached the highest accessible wave-number $$k_{\max} := (N/2)\Delta k.$$
In an inviscid ($\nu=0$) regime, energy injected at low wavenumbers eventually reaches $k_{\max}$ at a time scale of the order of $$T_\ast := k_{\max}/c.$$ 
We choose $\nu$ smaller than $ck_{\max}^{-3}$ to ensure a strong decrease of the spectral density at high wavenumbers (see Theorem \ref{heuristic-theorem_nu}). In practice, exploratory simulations will be carried out with various values of $\nu$, all of them satisfying this above criterion. This criterion is essential to ensure that no waves are reflected at the boundary of the artificial wavenumber periodic box.
We will consider that $T_\ast $ is the time when the statistically stationary regime is reached.
 
We let then the simulation run after  $T_\ast$, and we will average statistical estimators that we will define later over a set made of 100 instances every 
$(10^3 \Delta t)$ to ensure statistical independence\footnote{The number of instances and the time between each samples are empirically chosen.}. More explicitly, forthcoming time averages $\langle g[t]\rangle_t$ of a given function of time $g[t]$ will be taken as
\begin{equation}\label{time-average}
 \langle g[t] \rangle_t := \frac{1}{100} \sum_{n=1}^{100} g[T_\ast + n \times (10^3 \Delta t)].
\end{equation}
Accordingly, we stop the integration in time at $$T^\ast= T_\ast + 100 \times (10^3\Delta t).$$



With these choices, simulations are longer and longer as the number of collocation points $N$ increases, but it allows for smaller and smaller viscosities, which is necessary to develop an extended inertial range.

 \subsubsection{Definition and estimation of key statistical quantities}\label{RK:DefEstimators}

 
First of all, the $L^2$-norm $\sigma^2_u[t]$ of the solution $u[t,x]$ at a given instant $t$ is defined by
\begin{equation}\label{eq:L2norm}
\sigma^2_u[t] =\sum_k |\widehat{u}[t,k]|^2 \Delta k,
\end{equation}
where the sum is taken over all possible values of the discrete wave vector $k$.
Initially, we have $\sigma^2_u[0]=0$, it will then quickly grow and end up fluctuating around a certain average value way before $T_*$ (data not shown). This aforementioned mean value $\langle \sigma^2_u[t]\rangle_t$, where the time average procedure is defined in \eqref{time-average}, could be considered as the variance of the solution, and we have checked that it is independent of viscosity as expected if $\nu$ is chosen small enough. From the physical point of view, this is consistent with the observation that the variance of the solution of the Navier-Stokes equations get independent of viscosity at large Reynolds numbers \eqref{eq:AsymptVarianceNS}.

To characterize more precisely the statistical behavior of the solution when the $L^2$-norm starts fluctuating around a mean value \eqref{eq:L2norm}, we define the energy spectral density estimated as a periodogram, i.e. the norm square of the Fourier mode, that is
\begin{equation}\label{eq:PSD}
\widehat{C}_u(t,k) =  |\widehat{u}[t,k]|^2,
\end{equation}
and its averaged version
\begin{equation}\label{eq:PSDaverage}
\widehat{C}_u(k) =  \langle |\widehat{u}[t,k]|^2\rangle_t 
\end{equation}
where the time-average is defined in \eqref{time-average}. 


Another quantity of great importance is the so-called second-order structure function, i.e. the variance of the increment over a scale $\ell\in \R^d$, and given by
\begin{equation}\label{eq:S2L}
S_2(\ell) =  \langle (u[t,x+\ell]-u[t,x])^2\rangle_{t,x},
\end{equation}
where the time average is defined in \eqref{time-average} and the additional spatial-average is taken over $x\in\R^d$ for a given function of space $g[x]$: 
 $$
\langle g[x] \rangle_{x} := \frac{1}{L_{\rm tot}^d} \sum_{x}  g[x](\Delta x)^d.
$$
where the sum is taken over all possible values of the discrete wave vector $x$.
We recall that concerning solutions of the Navier-Stokes equations, it is observed that $S_2(\ell)$ behaves as $|\ell|^{2/3}$ at infinite Reynolds numbers, corresponding to a regularity of H\"{o}lder type with $H=1/3$ (c.f. \eqref{eq:AsymptVarianceIncrNS}).
  
 \subsubsection{Results and comments in dimension $d=1$}

 \begin{figure}
    \centering
 \includegraphics[width=14cm]{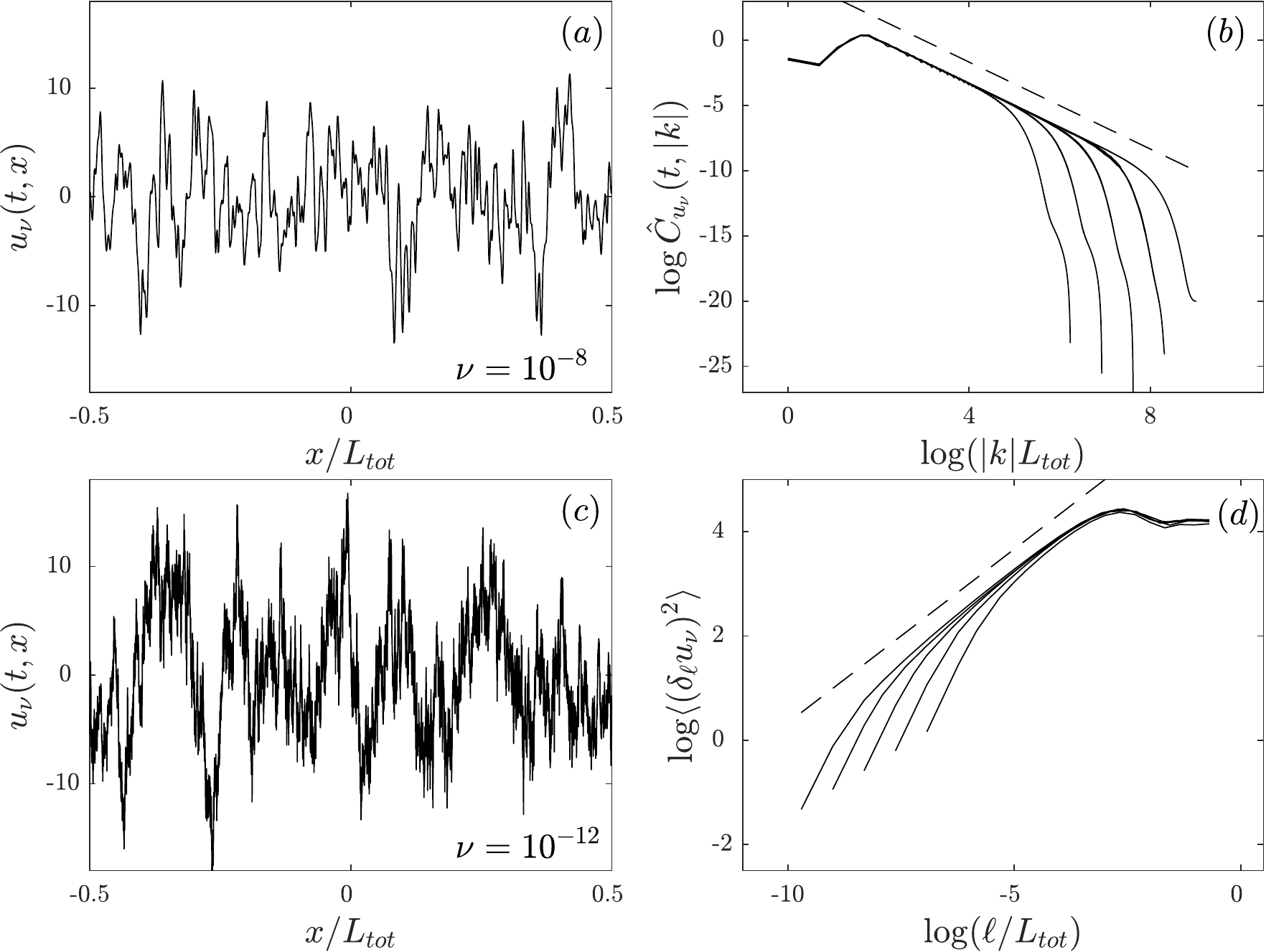}
    \caption{Local and statistical behaviors of the solution $u[t,x]$ to the evolution provided in paragraph \ref{Sec:Algo} for space dimension $d=1$, for a given viscosity $\nu$ in the statistically stationary range. All simulations have been done with $c=1$ and $H=1/3$ (a): spatial profiles of  $u[\cdot,x]$ at a given time $t$ in the statistically stationary regime for $\nu=10^{-8}$ obtained using $N=2^{10}$ collocation points.  (b): estimations of the energy spectral density based on the averaged periodograms (see \Cref{RK:DefEstimators} and \Cref{eq:PSDaverage}) of the solution for various values of viscosity $\nu=10^{-8}$, $10^{-9}$, $10^{-10}$, $10^{-11}$ and $10^{-12}$ (from left to right), using respectively $N=2^{10}$, $2^{11}$, $2^{12}$, $2^{13}$ and $2^{14}$  collocation points. We superimpose with a dashed line the asymptotic prediction $|k|^{-5/3}$. (c): same plot as in (a), but for a lower value of viscosity $\nu=10^{-12}$. (d): Similar plot as for (b) but for the second order structure function $S_2(\ell)$ \eqref{eq:S2L}, i.e. the variance of the increments, following an averaging procedure detailed in the text. We superimpose the expected asymptotic power-law behavior $\ell^{2/3}$. }\label{fig:1D}
  \end{figure}

We display in \Cref{fig:1D} the results of our simulations for space dimension $d=1$. We have used $c=1$ and $H=1/3$, and the time step $\Delta t=5 \times 10^{-3}$, whose value is motivated in section \ref{RK:Parameters}. As explained in \Cref{Sec:DeterViscAver}, we run the simulation until $T_*=k_{\max}/c$ at which all accessible length scales and wave lengths have been populated. After this transient, we then average various estimators such as the energy spectral density \eqref{eq:PSDaverage} and the second order structure function \eqref{eq:S2L} every 2 units of time. We indeed observe (data not shown) much before $T_*$ that the $L^2$-norm of the solution \eqref{eq:L2norm} fluctuates around a mean value. Typical snapshots of the solution $u[\cdot,x]$ are displayed in Figs. \ref{fig:1D}(a) and (c) at a time pertaining to the statistically stationary regime. A moderate viscosity $\nu=10^{-8}$ has been used in (a), whereas a much smaller one $\nu=10^{-12}$ is used in (c). We can see that as $\nu$ gets smaller, the velocity profile develops smaller scales and for the smallest viscosity that has been considered, it looks rougher. The averaged periodograms \eqref{eq:PSDaverage} obtained for all considered viscosities are represented in \Cref{fig:1D}(b). All estimated power spectral densities coincide at low wave lengths since the same forcing term has been used in all simulations. Then, at higher wave lengths, the spectra develop a universal power-law, independent of the characteristics of the forcing, with an exponent $-5/3$ which coincides with the expected $-(2H+d)$ exponent (see Eq.~\eqref{cor-u-heuristic-chi}) when choosing $H=1/3$. Then, at a characteristic wave length determined by the value of viscosity, spectra undergo a strong decrease, which is reminiscent of viscous damping. Similar behaviors are observed on the second order structure function $S_2(\ell)$ \eqref{eq:S2L}. At large length scales $\ell$ of the order of the scales where energy is injected into the system, determined by the spectral support of the forcing, $S_2(\ell)$ is independent of viscosity. At these scales, larger than the correlation length of the spatial profile of $u[t,x]$, $S_2(\ell)$ is approximatively equal to two times the variance of the solution. The fact that all curves superimposed at these scales show that this variance is indeed independent of viscosity. In the so-called inertial range of scales, as it is recalled in Section \ref{Sec:MainResults}, $S_2(\ell)$  develops a power-law behavior, whose exponent $2H$ governs the H\"{o}lder regularity of the asymptotic solution. At smaller length scales, viscosity smooths out any irregular variations such that  $S_2(\ell)$ can be Taylor expanded and becomes proportional to $\ell^2$.

 \begin{figure}
    \centering
 \includegraphics[width=14cm]{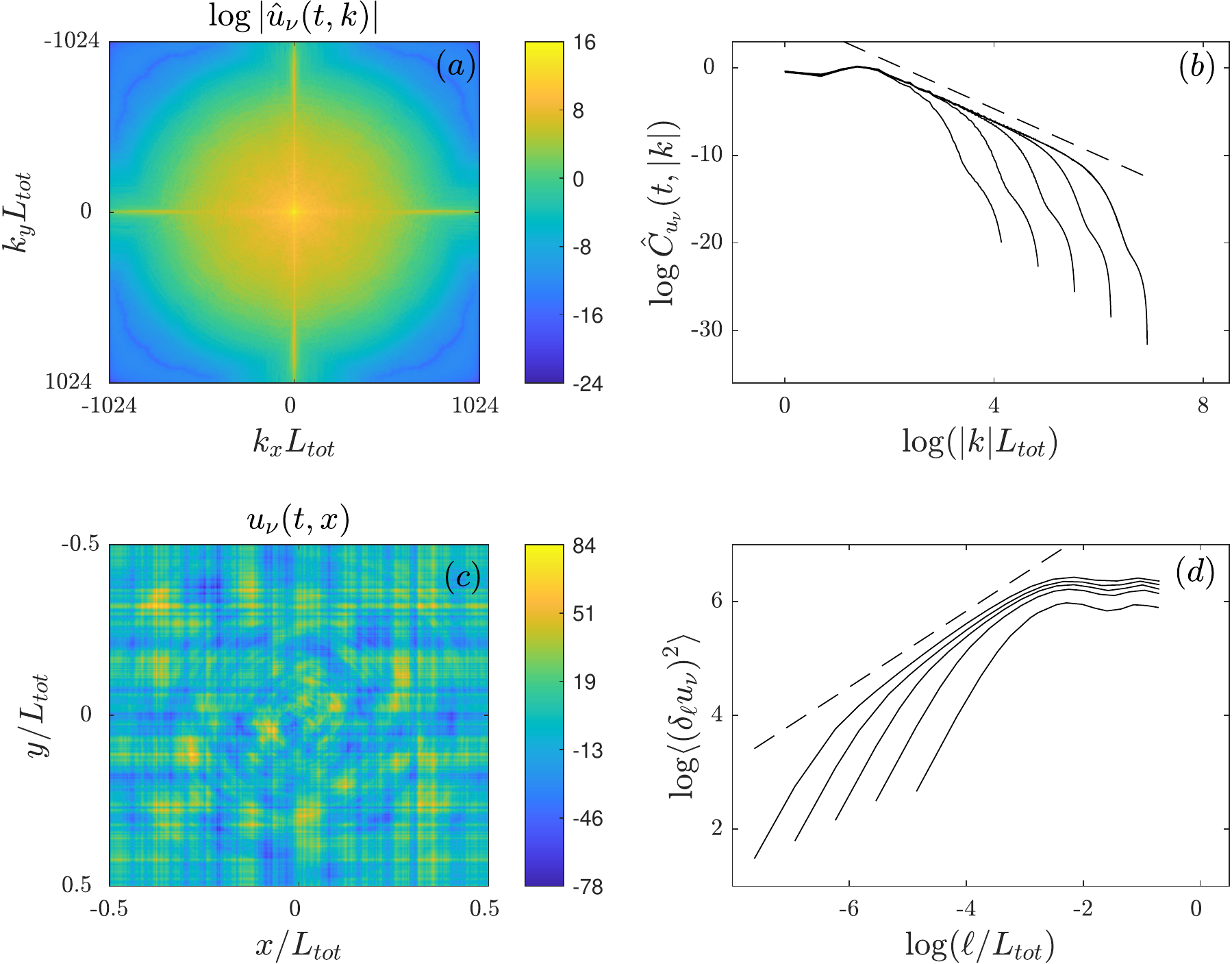}
    \caption{Local and statistical behaviors of the solution $u[t,x]$ to the evolution provided in paragraph \ref{Sec:Algo} for space dimension $d=2$, for a given viscosity $\nu$ in the statistically stationary range. All simulations have been done with $c=1$ and $H=1/3$ (a): Representation of the logarithm of the absolute value of  $\widehat{u}[\cdot,k]$ at a given time in the statistically stationary regime for $\nu=10^{-9}$ obtained using $N=2^{11}$ collocation points.  (b): estimations of the angle-averaged energy spectral density based on the averaged periodograms (see \Cref{RK:DefEstimators} and \eqref{eq:PSDaverage}) of the solution for various values of viscosity $\nu=10^{-5}$, $10^{-6}$, $10^{-7}$, $10^{-8}$ and $10^{-9}$ (from left to right), using respectively $N=2^{7}$, $2^{8}$, $2^{9}$, $2^{10}$ and $2^{11}$  collocation points. We superimpose with a dashed line the asymptotic prediction $|k|^{-2H-d}$, with $H=1/3$ and $d=2$. (c): similar plot as in (a) but for the corresponding spatial profiles of  $u[\cdot,x]$.  (d) Similar plot as for (b) but for the second order structure function $S_2(\ell)$ \eqref{eq:S2L}, i.e. the variance of the increments, following an averaging procedure detailed in the text. We superimpose the expected asymptotic power-law behavior $\ell^{2H}$, with $H=1/3$. }\label{fig:2D}
\end{figure}

 \subsubsection{Results and comments in dimension $d=2$}
 
We now display in \Cref{fig:2D} the results of our simulations for space dimension $d=2$. We have used again $c=1$ and $H=1/3$, and the time step $\Delta t=5 \times10^{-3}$. Similarly to the $d=1$ case, we propagate in time the discrete evolution provided in paragraph \ref{Sec:Algo} using the predictor-corrector method  until time $T_*$ at which all accessible scales have been populated. Before then, viscosity has damped all the energy at small scales such that the Fourier mode at $k_{\max}$ is exponentially small. Hence, we will consider that starting from time $T_*$, the statistically stationary regime has also been reached. We display in \Cref{fig:2D}(a) the amplitude of the Fourier modes $\widehat{u}[t,k]$ at a time $t$ lying in the statistically stationary regime, as a function of the two components $k_x$ and $k_y$ of the wave vector $k$, in a logarithmic representation. We can see that a lot of energy is concentrated along the two lines $(k_x,0)$ and $(0,k_y)$, whereas elsewhere in the plane, up to fluctuations, energy is distributed in a rotation invariant (i.e. isotropic) way. Recall that in a continuous framework, the statistical properties of Fourier modes $\widehat{u}(t,k)$ at any time $t$ are expected to depend only on the amplitude $|k|$. This shows that our numerical representation $\widehat{u}[t,k]$ is intrinsically anisotropic. We interpret this spurious anisotropy as the consequence of the finiteness of our simulation domain, with the implied finiteness of the resolution $\Delta k=1/L_{tot}$, but also the fact that the Fourier modes $\widehat{u}[t,k]$ are distributed on a Cartesian grid, whereas the continuous solution $\widehat{u}(t,k)$ is expected in average to be spherically symmetric. Another limitation of our numerical approach is related to the rough nature of the expected solution. It is clear from the inspection of \Cref{fig:2D}(a) that Fourier modes are correlated and smoother than what is expected from the Fourier transform of a statistically isotropic random field. This is very possibly related to the estimation of the divergence operator entering in the evolution \eqref{eq:maineq-div-viscous} with finite differences, as it is implicitly done using back and forth DFTs (see \Cref{eq:MeaningDifferentK}). A specifically devoted article aimed at exploring the numerical representation of the continuous formulation provided in \eqref{eq:maineq-div-viscous} would be needed, designing for instance some finite volume algorithms able to deal with the rough nature of the underlying fields. It will be the subject of future publications. 

Accordingly, the representation of the solution in physical space, that we display in \Cref{fig:2D}(c), exhibits two types of anisotropies. The first type of anisotropy can be observed along the two directions $x$ and $y$ of the Cartesian frame as straight lines. This anisotropy is consistent with what is observed on the Fourier transform displayed in \Cref{fig:2D}(a) along $k_x$ and $k_y$. From the inspection of \Cref{fig:2D}(c), another type of anisotropy can be evidenced around the origin, i.e. around $x=y=0$. Recall that for a statistically homogenous field, probability laws are expected to be invariant by translation, and thus the origin of the domain should not play a particular role. In our numerical solution, this is clearly not the case, and we believe that this anisotropy is related to the aforementioned correlated nature of the Fourier modes displayed in \Cref{fig:2D}(a). Once again, more work is needed to design proper numerical schemes able to get rid of these spurious anisotropies.

Nonetheless, despite these anisotropies, that are not present in the solution of the continuous framework, our numerical solution behaves in a statistical sense in the expected way. For instance, we display in \Cref{fig:2D}(b) the estimation of the power spectral densities, obtained while averaging in time the square of the amplitude of Fourier modes. We underline that these spectral densities are furthermore averaged over the angles that the wave vector $k$ makes with the axes of the Cartesian frame, such that displayed spectral densities depend only on the wave length amplitude $|k|$. We indeed observe that, as viscosity gets smaller and smaller, spectral densities develop a power-law behavior in the inertial range of scales, with the expected exponent $-(2H+d)$ which is derived in \eqref{cor-u-heuristic-chi}. At higher wavenumbers, Fourier modes are exponentially damped by viscosity. Similar conclusions could be drawn from inspection of the scale dependence of the second order structure function $S_2(\ell)$ that is shown in \Cref{fig:2D}(d). Let us mention here that spatial averages have been obtained over the full spatial domain, averaging furthermore over two structure functions obtained while considering spatial lags $\ell=|\ell|e_x$ and $\ell=|\ell|e_y$, where $e_x$  and $e_y$ are the two orthonormal unit vectors of the Cartesian frame. At the largest scales, above the characteristic ones of the forcing, we can see that $S_2(\ell)$ reaches a plateau, which gets independent of viscosity as $\nu\to 0$. Similarly to the $d=1$ case, above the correlation length of $u[t,x]$, $S_2(\ell)$ coincides with two times the variance of $u[t,x]$, which says in other words that the variance gets itself independent of viscosity if $\nu$ is chosen small enough. At lower scales, i.e. in the inertial range, $S_2(\ell)$ develops a power-law behavior of exponent $2H$, as it is expected from the predictions made in the continuous framework (\Cref{thm:Holder_sol}), over a range of scales which grows as viscosity gets smaller and smaller. Finally, at even smaller scales, viscous effects dominate and smooth out the spatial profiles, such that  $S_2(\ell)$ gets proportional to $\ell^2$.

 \begin{figure}[t]
    \centering
 \includegraphics[width=14cm]{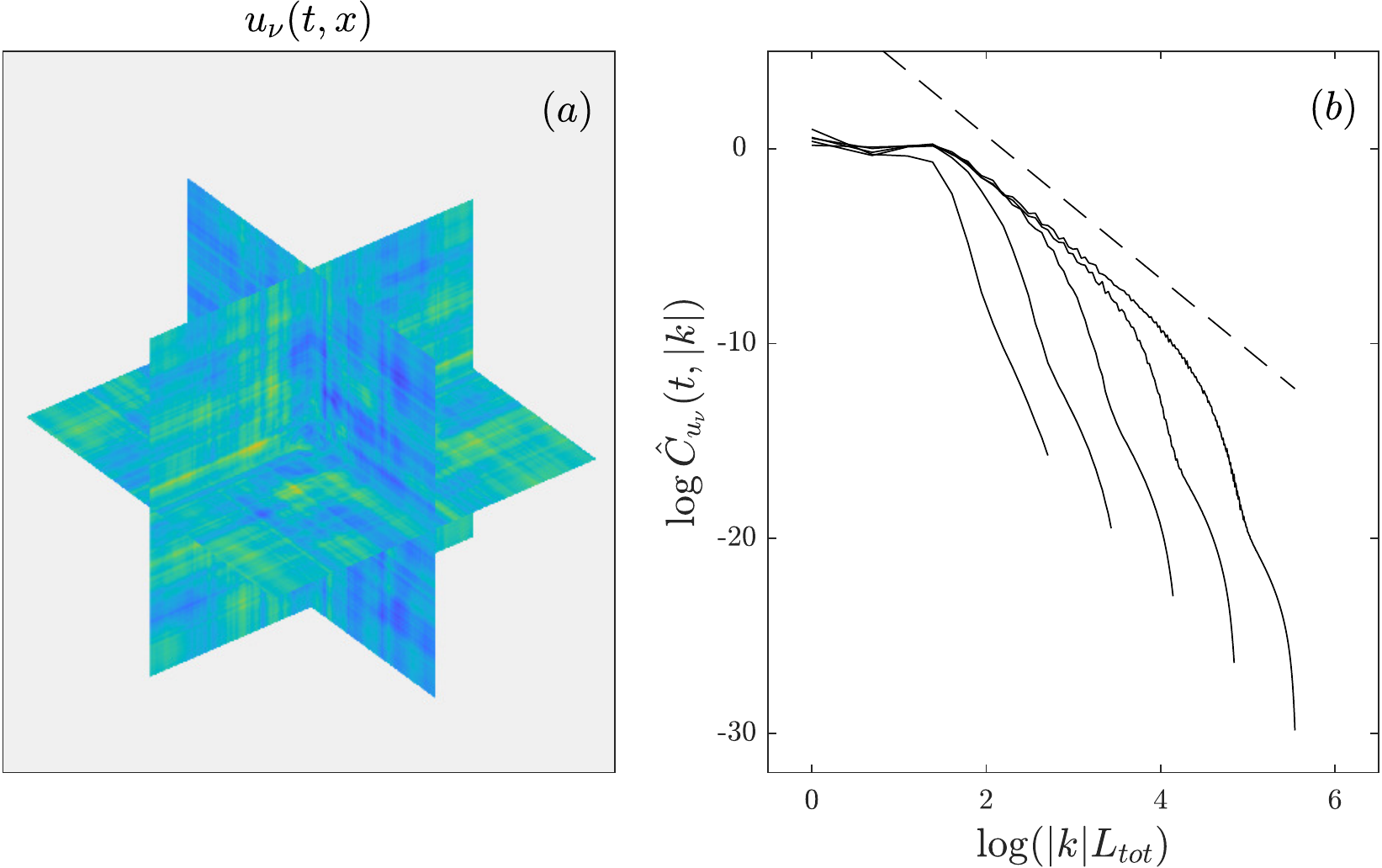}
    \caption{Local and statistical behaviors of the solution $u[t,x]$ to the evolution provided in paragraph \ref{Sec:Algo} for space dimension $d=3$, for a given viscosity $\nu$ in the statistically stationary range. All simulations have been done with $c=1$ and $H=1/3$ (a): Representation of some slices  of  $u[\cdot,x]$ at a given time in the statistically stationary regime for $\nu=10^{-7}$ obtained using $N=2^{9}$ collocation points in each directions.  (b): estimations of the angle-averaged energy spectral density based on the averaged periodograms (see \Cref{RK:DefEstimators} and \eqref{eq:PSDaverage}) of the solution for various values of viscosity $\nu=10^{-3}$, $10^{-4}$, $10^{-5}$, $10^{-6}$ and $10^{-7}$ (from left to right), using respectively $N=2^{5}$, $2^{6}$, $2^{7}$, $2^{8}$ and $2^{9}$  collocation points. We superimpose with a dashed line the asymptotic prediction $|k|^{-2H-d}$, with $H=1/3$ and $d=3$.  }\label{fig:3D}
  \end{figure}

 \subsubsection{Results and comments in dimension $d=3$}
Let us now finish this Section devoted to numerical simulations by presenting in \Cref{fig:3D} the results for space dimension $d=3$. Let us mention that in a three-dimensional setup, simulations are much more demanding from a computational perspective, and the cost of performing back and forth FFTs gets much higher, because derivatives along the 3 directions have to be considered, and also because the number of discretization points, $N^3$, increase tremendously as $N$ increases. For these reasons, and because we are propagating the integration in time from a vanishing initial condition towards the statistically stationary regime before taking averages, we could not go above $N=2^9=512$ collocation points along each direction. Consequently, we have not been able to perform simulations for viscosities smaller than $\nu=10^{-7}$. Nonetheless, we observe (data not shown) that fluctuations as quantified by the $L^2$-norm \eqref{eq:L2norm} get independent of viscosity in a good approximation starting from $\nu=10^{-6}$, as it is expected from the behavior of the solution in a continuous framework. As mentioned, and similarly to the $d=1$ and $d=2$ cases, we go through the transient while integrating the solution until time $T_*$, and only then we start taking averages.

We display in \Cref{fig:3D}(a) a rendering of our three-dimensional simulations in physical space in the Cartesian frame. For the sake of clarity, and because visualizations gets more complicated, we only show three slices along the planes $(x,y,0)$, $(x,0,z)$ and $(0,y,z)$. Bright and dark colors correspond respectively to large positive and large negative fluctuations of the solution, similarly to what has been observed for the $d=2$ case, which is displayed in \Cref{fig:2D}(c). Once again, we observe the two types of anisotropies that we evidenced in the $d=2$ case, one along the three directions of the Cartesian frame, and one around the origin. 

Nonetheless, as it is displayed in  \Cref{fig:3D}(b), the simulations behave as expected in a statistical manner, as it can be observed on the power spectral densities \eqref{eq:PSDaverage}. Once again, these densities are not only averaged in time once the statistically stationary regime has been reached, but they are also averaged over the two angles that the wave vector $k$ does with the Cartesian axes, such that $\widehat{C}_u(k)$ is a function of the norm $|k|$ only. We can see that at large scales, i.e. at low wavenumbers $|k|$, fluctuations are independent of viscosity, even for the highest value $\nu=10^{-3}$ which is used for the lowest number of collocation points $N=2^5$. It is nonetheless crucial to consider smaller values of viscosity, necessitating thus higher values of $N$, up to $N=2^9$, in order to develop an extended inertial range. It is clear for the smallest value of viscosity ($\nu=10^{-7}$ for $N=2^9$) that the spectral density has developed a power-law behavior in the inertial range, with the expected exponent $-(2H+d)$, in a consistent manner with our theoretical predictions \eqref{cor-u-heuristic-chi}.

We thus see that we are able to give an appropriate numerical representation of the continuous framework using the DFT to define the divergence operator entering in the spectral evolution provided in \eqref{eq:maineq-div-viscous}. In particular, we are able to reproduce in a discrete setup the statistical behaviors of the spectral densities (Eq.~\ref{cor-u-heuristic-chi})  and second-order structure functions (\Cref{thm:Holder_sol}), and their related power-law behaviors. Nonetheless, more work is needed to get rid of the anisotropies that are clearly observed for the $d=2$ and $d=3$ cases. To do so, a promising direction could be given while designing a finite volume scheme able to respect the inherent spherical symmetry of the deterministic part of the evolution 
\eqref{eq:maineq-div-viscous}. This is required in order to propose a realistic model of fully developed fluid turbulence that could be used in an efficient way in various applications, in which spatial fluctuations of the velocity field have crucial consequences on the evolutions of dynamical quantities of interest. We keep these important developments for future investigations.

\bibliographystyle{hsiam}
\bibliography{references}
\end{document}